\documentclass[letterpaper, 11pt]{article}

\usepackage[margin=1in]{geometry}

\usepackage[dvipsnames,usenames]{xcolor}
\usepackage[colorlinks=true,pdfpagemode=UseNone,urlcolor=RoyalBlue,linkcolor=RoyalBlue,citecolor=OliveGreen,pdfstartview=FitH]{hyperref}

\usepackage{mathtools, amssymb, amsmath, amsthm}
\usepackage{thmtools}
\declaretheorem{theorem}

\declaretheorem[sibling=theorem]{lemma}
\declaretheorem{example}
\declaretheorem[sibling=theorem]{corollary}

\usepackage{bm}
\usepackage{yhmath}
\usepackage{tikz}
\usetikzlibrary{shapes}
\usepackage{subcaption}
\usepackage[T1]{fontenc}

\usepackage[numbers]{natbib}

\usepackage{todonotes}

\newcommand{\pointset}{\Omega}
\newcommand{\config}{X}
\newcommand{\R}{\mathbb{R}}

\newcommand{\ALG}{\text{\rm ALG}}
\newcommand{\OPT}{\text{\rm OPT}}

\DeclareMathOperator*{\argmin}{arg\,min}
\DeclareMathOperator*{\argmax}{arg\,max}

\title{
    Deterministic $3$-Server on a Circle and the Limitation of Canonical Potentials
}

\author{
    Zhiyi Huang
    \thanks{The University of Hong Kong. Email: zhiyi@cs.hku.hk.}
    \and
    Hanwen Zhang
    \thanks{IIIS, Tsinghua University. Email: zhanghw18@mails.tsinghua.edu.cn.}
}

\date{May 2022}

\begin{document}

\begin{titlepage}
    \thispagestyle{empty}
    \maketitle
    \begin{abstract}
        The deterministic $k$-server conjecture states that there is a $k$-competitive deterministic algorithm for the $k$-server problem for any metric space.
We show that the work function algorithm is $3$-competitive for the $3$-server problem on circle metrics, a case left open by Coester and Koutsoupias (2021).
Our analysis follows the existing framework but introduces a new potential function which may be viewed as a relaxation of the counterpart by Coester and Koutsoupias (2021).
We further notice that the new potential function and many existing ones can be rewritten in a canonical form.
Through a computer-aided verification, however, we find that no such canonical potential function can resolve the deterministic $3$-server conjecture for general metric spaces under the current analysis framework.
    \end{abstract}
\end{titlepage}

\section{Introduction}
\label{sec:introduction}

We study the $k$-server problem introduced by \citet*{ManasseMS:STOC:1998}.
Consider a metric space and $k$ servers located in the space.
In each round a point in the metric makes a request.
If currently no server is at that point then the online algorithm needs to move a server there in this round.
We want to minimize the total distance that the servers travel to serve the requests.
Following the standard competitive analysis, we compare the algorithm's total distance to the minimum total distance in hindsight.
The maximum of this ratio over all possible instances is the \emph{competitive ratio} of the algorithm.

The $k$-server problem is one of the oldest and hardest online algorithm problems, capturing online caching and paging as special cases.
It offers two long-standing open questions that have been an endless source of inspiration to online algorithm researchers:

\begin{quote}
    \textbf{Deterministic $k$-server conjecture:}\\
    \emph{For any metric space there is a deterministic $k$-competitive algorithm.}
\end{quote}

\begin{quote}
    \textbf{Randomized $k$-server conjecture:}\\
    \emph{For any metric space there is a randomized $O(\log k)$-competitive algorithm.}
\end{quote}

This paper focuses on deterministic algorithms.
On the one hand, \citet{ManasseMS:STOC:1998} showed that any deterministic algorithm's competitive ratio is at least $k$.
On the other hand, \citet{KoutsoupiasP:JACM:1995} proposed the \emph{work function algorithm} (WFA), and proved that the WFA is $(2k-1)$-competitive for any metric space.
Although it is widely conjectured that the WFA is in fact $k$-competitive, the $2k-1$ ratio remains the best known bound for general metric spaces after more than two decades of subsequent efforts.

Since a direct attack to the deterministic $k$-server conjecture on general metrics appears beyond reach, a lot of efforts have been devoted to special metric spaces.
For example, \citet{ManasseMS:STOC:1998} resolved the case of $k=2$ servers and the case when the metric space has only $k+1$ points;
\citet*{ChrobakKPV:SIDMA:1991} settled the case of line metrics;
\citet{ChrobakL:SICOMP:1991} proved the case of tree metrics;
\citet*{BeinCL:TCS:2002} solved the case of $k=3$ on a Manhattan plane;
and \citet{BartalK:TCS:2004} worked out the case when the metric has only $k+2$ points and the case of weighted star metrics (a.k.a., the weighted paging problem).

The most recent progress was due to \citet{CoesterK:ICALP:2021}, who gave a unified potential function covering all known metric spaces for which the WFA is known to be $k$-competitive.
They highlighted two cases as the frontier of the deterministic $k$-server conjecture: $k=3$ (because $k=2$ is known), and circle metrics (because they are the simplest non-tree metrics).
They further used the combination of the two, i.e., $3$ servers on circle metrics, to refute a lazy-adversary conjecture by \citet{ChrobakL:SICOMP:1991}.
They pointed out that their unified potential fails in the case of $3$ servers on circle metrics, and questioned whether the WFA is $3$-competitive in this case and even the correctness of the deterministic $k$-server conjecture.%
\footnote{Coester and Koutsoupias questioned the correctness of the deterministic $3$-server conjecture in their paper, and Coester asked whether the WFA is $3$-competitive for $3$ servers on circle metrics in the conference talk.}

\subsection{Our Contributions}

Our contributions are twofold.
We introduce a new potential function and use it to prove that the WFA is $3$-competitive for the $3$-server problem on circle metrics, answering affirmatively the question left by \citet{CoesterK:ICALP:2021}.
The new potential may be viewed as a relaxation of theirs.
Our analysis still follows the existing framework, building on the known properties of the work functions, including quasi-convexity, duality, and the extended cost lemma.
To this end, our result reconsolidates the belief that the WFA is $k$-competitive in general, optimistically even provable under the existing framework.
See Section~\ref{sec:3-server}.

On the flip side, we define a canonical family of potential functions that captures as special cases the new potential function in this paper, and the previous potential functions of 
\citet{chrobak1991server}, \citet{BeinCL:TCS:2002}, and \citet{CoesterK:ICALP:2021}. 
We further provide evidence that this canonical family of potential functions cannot resolve the $k=3$ case on general metrics.
Hence, fundamentally new ideas are still needed for further advancement on the deterministic $k$-server conjecture.
See Section~\ref{sec:potential}.


\subsection{Further Related Works}

There has been lots of progresses on the randomized $k$-server conjecture.
\citet*{BansalBN:JACM:2012} solved the case of weighted star metrics (a.k.a., the weighted paging problem).
For general metric spaces, a breakthrough by \citet*{BansalBMN:FOCS:2011} gave a $\mbox{polylog}(k,n)$-competitive algorithm, where $n$ is the number of points in the metric space.
Subsequently, \citet{BubeckCLLM:STOC:2018} gave an $O(\log^2 k)$-competitive algorithm for hierarchically separated trees (HST), which implies an improved $O(\log^2 k \log n)$-competitive algorithm for general metrics via Bartal's embedding (c.f., Bartal~\cite{Bartal:FOCS:1996, Bartal:STOC:1998} and \citet*{FakcharoenpholRT:JCSS:2004}), or an $O(\log^3 k \log \Delta)$-competitive algorithm through their own embedding, where $\Delta$ is the ratio between the maximum and minimum distances between pairs of points.
Their algorithm and analysis were later simplified by \citet*{BuchbinderGMN:SODA:2019}.

The WFA has found applications in other online algorithm problems, notably in several generalziations of the $k$-server problem.
\citet*{BorodinLS:JACM:1992} proved that the WFA achieves the optimal competitive ratio in metrical task systems.
\citet{Sitters:SICOMP:2014} showed that the WFA is constant competitive for the generalized $2$-server problem.
\citet*{BansalEK:FOCS:2017} used the WFA to obtain a near optimal competitive ratio in the weighted $k$-server problem.
Most recently, \citet*{ArgueGTG:JACM:2021} and \citet{Sellke:SODA:2020} applied the concept of work function to chasing convex bodies and got near optimal competitive ratios.




\section{Preliminaries}
\label{sec:preliminary}

Consider a metric space $(\pointset, d)$, where $\pointset$ is a set of points and $d$ is a metric on $\pointset$.
In other words, $d(x, y)$ is the distance between two points $x, y \in \pointset$. 
Let $\pointset^k$ denote the set of multi-sets of $k$ points from $\pointset$.
We will abuse notation and let $d(X, Y)$ also denote the Wasserstein distance between two finite multi-sets $X, Y \in \pointset^k$.
We further let $x^k$ denote the multi-set of $k$ copies of $x \in \pointset$, and let $x_1 x_2 \dots x_k$ denote the multi-set of $k$ points $x_1, x_2, \dots, x_k \in \pointset$. 

An instance of the $k$-server problem is defined by an initial configuration $\config_0 \in \pointset^k$, and a sequence of requests $r_1, r_2, \dots, r_T$.
The initial configuration is given to the algorithm at the beginning.
Then, the algorithm receives the requests one at a time.
Upon receiving each request $r_t$, the algorithm immediately chooses a configuration $\config_t \in \pointset^k$ with $r_t \in \config_t$.
The objective is to minimize the total distance between neighboring configurations, i.e.:
\[
    \sum_{t=1}^T d \big( \config_{t-1}, \config_t \big)
    ~.
\]

Let $\ALG$ denote the objective given by the algorithm's configurations.
Following the standard worst-case competitive analysis, we compare to the minimum cost in hindsight, i.e.:
\[
    \OPT ~=~ \min_{(\config_1, \dots, \config_T) ~:~ \forall 1 \le t \le T, r_t \in \config_t} ~ \sum_{t=1}^T d \big( \config_{t-1}, \config_t \big)
    ~.
\]

The competitive ratio of an online algorithm with respect to (w.r.t.) a given metric space is the smallest $\Gamma$ for which $\ALG \le \Gamma \cdot \OPT + O(1)$ for all instances of the $k$-server problem on this metric, where $O(1)$ represents a term that does not depend on the number of requests $T$ (but may depend on the number of servers $k$ and the diameter of the metric).

It will be convenient to further consider the minimum cost up to any time $1 \le t \le T$, and \emph{with a given final configuration}.
For any time $1 \le t \le T$, and any configuration $\config_t \in \pointset^k$ with $r_t \in \config_t$, define:
\[
    \OPT_t \big( \config_t \big) ~=~ \min_{(\config_1, \dots, \config_{t-1}) ~:~ \forall 1 \le i \le t-1, r_i \in \config_i} ~ \sum_{i=1}^t d\big(\config_{i-1}, \config_i\big)
    ~.
\]

\subsection{Work Functions and the Work Function Algorithm}

The work function $w_t : \pointset^k \mapsto \R $ at any time $0 \le t \le T$ is the optimal objective for processing the first $t$ requests, and \emph{further moving to a given final configuration}, i.e.:
\[
    w_t\big(\config\big) ~=~ \min_{(\config_1, \dots, \config_t) ~:~ \forall 1 \le i \le t, r_i \in \config_i} ~ \sum_{i=1}^t d \big( \config_{i-1}, \config_i \big) + d \big( \config_t, \config \big)
    ~.
\]

For $1 \le t \le T$, this is equivalent to:
\[
    w_t\big(\config\big) ~=~ \min_{\config_t \in \pointset^k ~:~ r_t \in \config_t} \OPT_t\big(\config_t\big) + d\big(\config_t, \config\big)
    ~.
\]

The work function algorithm (WFA) selects at each time $1 \le t \le T$ the minimizer in the above definition of the work function at time $t$ when the input is $\config_{t-1}$, i.e.:
\[
    \config_t \in \argmin_{\config \in \pointset^k ~:~ r_t \in \config} \OPT_t\big(\config\big) + d\big(\config, \config_{t-1}\big)
    ~.
\]

\begin{lemma}[c.f., \citet{koutsoupias2009k}]
    \label{lem:work-function}
    The work functions satisfy the following properties.
    
    \begin{enumerate}
        \item (Monotonicity)
            For any $1 \le t \le T$, and any $X \in \pointset^k$:
            \[
                w_t(X) \ge w_{t-1}(X)
                ~,
            \]
            and it holds with equality when $r_t \in X$.
        \item (Lipschitzness)
            For any $0 \le t \le T$, and any $X, Y \in \pointset^k$:
            \[
                w_t\big(X\big) - w_t\big(Y\big) \le d\big(X, Y\big)
                ~.
            \]

        \item (Quasi-convexity)
            For any $0 \le t \le T$, and any multi-sets $X, Y \in \pointset^k$, there exists a bijection $\pi : X \to Y$,%
            \footnote{The bijection treats copies of the same point as distinct objects, and may map them differently in general.}
            such that for any partition $X = X_1 \sqcup X_2$:
            \[
                w_t \big(X\big) + w_t \big(Y\big) \ge w_t \big(X_1 \cup \pi(X_2) \big) + w_t \big(\pi(X_1) \cup X_2 \big)
                ~.
            \]
            Further, the bijection $\pi$ maps the common points in $X, Y$ to themselves.%
            \footnote{That is, if an point $x$ appears $n_X$ and $n_Y$ times in $X$ and $Y$ respectively, $\pi$ maps $\min \{ n_X, n_Y \}$ copies of $x$ to $x$.}
        \item (Duality)
            \footnote{The original duality also asserts $\argmin_{X \in \pointset^k} \big( w_{t-1}(X) - d(X, r_t^k) \big) \subseteq \argmin_{X \in \pointset^k} \big( w_t(X) - d(X, r_t^k) \big)$. It is not needed in this paper because we focus on metrics with antipodes including the circle metric and general metrics (through appropriate extension), and because of Lemma~\ref{lem:antipode} in the next section.}
            For any $1 \le t \le T$:
            %
            %
            %
            \[
                \argmin_{X \in \pointset^k} \big( w_{t-1}(X) - d(X, r_t^k) \big) \subseteq \argmax_{X \in \pointset^k} \big( w_t(X) - w_{t-1}(X) \big)
                ~.
            \]
        \item (Extended Cost Lemma)
            Suppose that for any $k$-server instance:
            \[
                \sum_{t=1}^T \max_{X \in \pointset^k} \big( w_t(X) - w_{t-1}(X) \big) \le (\Gamma+1)\,\OPT + O(1)
                ~.
            \]
            Then, the WFA is $\Gamma$-competitive.
    \end{enumerate}
\end{lemma}



\subsection{Potential-based Approach}

Most existing analyses of the WFA establish the condition of the Extended Cost Lemma (i.e., the last part of Lemma~\ref{lem:work-function}) by designing appropriate potentials $\Phi_0, \Phi_1, \dots, \Phi_T$ such that:
\begin{align}
    \label{eqn:potential-approach-stepwise}
    \Phi_t - \Phi_{t-1} & ~ \ge ~ \max_{X \in \pointset^k} \big( w_t(X) - w_{t-1}(X) \big)
    ~, && \forall 1 \le t \le T; \\
    \label{eqn:potential-approach-bounded}
    \Phi_T - \Phi_0 ~~ & ~ \le ~ (\Gamma+1)\,\OPT + O(1)
    ~.
\end{align}

Specifically, $\Phi_t$ could be the sum of work function $w_t$'s values on $\Gamma+1$ configurations, plus some other bounded terms.
Then, the boundedness of the potentials implies the second condition above.

\section{Deterministic \texorpdfstring{$3$-Server}{3-Server} on a Circle}
\label{sec:3-server}

This section considers the circle metrics and $k=3$ servers.
Let $\Delta > 0$ denote the diameter of a circle metric.
We may write the set of points as an interval $\pointset = [0, 2\Delta)$.
The corresponding metric is $d(x, y) = \min \{ |x-y|, 2\Delta-|x-y| \}$ for any points $x, y \in \pointset$.

\subparagraph*{Antipodes.}
For any $x \in \pointset$, define its antipode as $\bar{x} = x + \Delta$ if $x \in [0, \Delta)$ and $\bar{x} = x - \Delta$ if $x \in [\Delta, 2\Delta)$.
The antipodes satisfy that for any $x, y \in \pointset$:
\begin{equation}
    \label{eqn:antipode}
    d(x, y) + d(\bar{x}, y) = \Delta
    ~.
\end{equation}

\subparagraph{Arcs.}
For any $x, y \in \pointset$, denote the arc connecting them as $\wideparen{xy}$.
Define $\wideparen{x\bar{x}} = \pointset$.
For any $z \in \wideparen{xy}$, we have:
\[
    d(x, z) + d(z, y) = d(x, y)
    ~.
\]

The next lemma follows from the Lipschitzness of work functions (Lemma~\ref{lem:work-function}).

\begin{lemma}[c.f., \citet{Koutsoupias:FOCS:1999}]
    \label{lem:antipode}
    For any $1 \le t \le T$:
    \[
        \bar{r}_t^k \in \argmin_{X \in \pointset^k} \big( w_{t-1}(X) - d(X, r_t^k) \big)
        ~.
    \]
\end{lemma}

\begin{theorem}
    \label{thm:3-server}
    The WFA is $3$-competitive for $3$-server on a circle.
\end{theorem}

\begin{proof}
    For any work function $w_t$ and any points $u, x, y, z \in \pointset$, define:
    \begin{equation}
        \label{eqn:potential}
        \Phi(w_t, u, x, y, z) = w_t(\bar{u}^3) + w_t(ux\bar{y}) + w_t(uy\bar{z}) + w_t(uz\bar{x}) - d(x, y) - d(y, z) - d(z, x)
        ~.
    \end{equation}



    We remark that function $\Phi$ by definition is invariant to permuting $x, y, z$ while keeping the same parity.
    It is also invariant to changing the parity of the permutation of $x, y, z$ \emph{and} changing them into their antipodes $\bar{x}, \bar{y}, \bar{z}$ respectively, since Eqn.~\eqref{eqn:antipode} implies: 
    \begin{align*}
        \Phi(w_t, u, x, y, z)
        &
        = w_t(\bar{u}^3) + w_t(ux\bar{y}) + w_t(uy\bar{z}) + w_t(uz\bar{x}) - d(x, y) - d(y, z) - d(z, x) \\
        &
        = w_t(\bar{u}^3) + w_t(ux\bar{y}) + w_t(uy\bar{z}) + w_t(uz\bar{x}) - d(\bar{x}, \bar{y}) - d(\bar{y}, \bar{z}) - d(\bar{z}, \bar{x}) \\
        &
        = \Phi(w_t, u, \bar{y}, \bar{x}, \bar{z})
        ~.
    \end{align*}


    For any $0 \le t \le T$, define the potential after request $r_t$ to be:
    \begin{equation}
        \label{eqn:potential-t}
        \Phi_t = \min_{u, x, y, z \in \pointset} \Phi(w_t, u, x, y, z)
        ~.
    \end{equation}
    
    The next lemma characterizes the minimizer in the above definition of potentials. 
    We sketch its proof at the end of the section, and defer the complete proof to Appendix~\ref{app:proof-3-server}.
    
    \medskip

\begin{lemma}
    \label{lem:3-server}
    For any $3$-sever instance on a circle metric, and any time $t$, there are points $x_t, y_t, z_t \in \pointset$ such that $\Phi_t = \Phi(r_t, x_t, y_t, z_t)$.
\end{lemma}

    \medskip

    On the one hand, by $w_T(X) \le \OPT + 3\Delta$ for any $X \in \pointset^3$ and $d(x, y), d(y, z), d(z, x) \ge 0$:
    \[
        \Phi_T \le 4 \OPT + 12\Delta
        ~.
    \]

    On the other hand, by $w_0(X) \ge 0$ for any $X \in \pointset^3$ and by $d(x, y), d(y, z), d(z, x) \le \Delta$:
    \[
        \Phi_0 \ge -3\Delta
        ~.
    \]

    Hence:
    \[
        \Phi_T - \Phi_0 \le 4 \OPT + 15 \Delta = 4 \OPT + O(1)
        ~.
    \]

    Finally, let $x_t, y_t, z_t \in \pointset$ be the points from Lemma~\ref{lem:3-server}.
    We have:
    \begin{align*}
        \Phi_t - \Phi_{t-1}
        &
        = \Phi(w_t, r_t, x_t, y_t, z_t) - \min_{u, x, y, z \in \pointset} \Phi(w_{t-1}, u, x, y, z) \\
        &
        \ge \Phi(w_t, r_t, x_t, y_t, z_t) - \Phi(w_{t-1}, r_t, x_t, y_t, z_t) \\[1.5ex]
        &
        = \big( w_t(\bar{r}_t^3) - w_{t-1}(\bar{r}_t^3) \big) + \big( w_t(r_t x_t \bar{y}_t) - w_{t-1}(r_t x_t \bar{y}_t) \big) \\
        & \quad
        + \big( w_t(r_t y_t \bar{z}_t) - w_{t-1}(r_t y_t \bar{z}_t) \big) + \big( w_t(r_t z_t \bar{x}_t) - w_{t-1}(r_t z_t \bar{x}_t) \big)
        \\[1ex]
        & \ge w_t(\bar{r}_t^3) - w_{t-1}(\bar{r}_t^3)
        \tag{Monotonicity in Lemma~\ref{lem:work-function}} \\[1.5ex]
        &
        = \max_{X \in \Omega^k} \big( w_t(X) - w_{t-1}(X) \big)
        ~.
        \tag{Lemma~\ref{lem:antipode}, and Duality in Lemma~\ref{lem:work-function}}
    \end{align*}

    The theorem now follows by the Extended Cost Lemma in Lemma~\ref{lem:work-function}.
\end{proof}

To prove Lemma~\ref{lem:3-server}, it is more convenient to consider an equivalent form of our potentials defined in Equations~\eqref{eqn:potential} and \eqref{eqn:potential-t}.%
\footnote{We choose the form in Eqn.~\eqref{eqn:potential} for consistency with the discussion in Section~\ref{sec:potential}.}
Define:
\begin{equation}
    \label{eqn:potential-alternative}
    \Phi^\star(w_t, u, x, y, z) = w_t(\bar{u}^3) + w_t(ux\bar{y}) + w_t(uy\bar{z}) + w_t(uz\bar{x})
    ~.
\end{equation}

Further let $P$ denote the set of point tuples $(x, y, z) \in \pointset^3$ that are \emph{not} in the interior of the same semi-circle:
%
\[
    P = \big\{ (x, y, z) \in \pointset^3 \mid d(x,y) + d(y,z) + d(z,x) = 2\Delta \big\}
    ~.
\]

%

We establish the equivalence through the next lemma.

\begin{lemma}
    \label{lem:potential-equivalence}
    For any $1 \le t \le T$:
    \[
        \Phi_t = \min_{u \in \pointset, (x, y, z) \in P} \Phi^\star(w_t, u, x, y, z) - 2\Delta
        ~.
    \]
\end{lemma}

\begin{proof}
    By the definitions of $\Phi$ and $\Phi^\star$, we have:
    \[
        \Phi_t \le \min_{u \in \pointset, (x, y, z) \in P} \Phi^\star(w_t, u, x, y, z) - 2\Delta
        ~.
    \]

    It remains to show the opposite direction.
    Suppose that $\Phi_t = \Phi(w_t, u_t, x_t, y_t, z_t)$.
    If $(x_t, y_t, z_t) \in P$ the lemma holds by the definition of $P$.

    Next suppose that $(x_t, y_t, z_t) \notin P$.
    By the symmetry of $x, y, z$ in the definition of $\Phi$, we may assume without loss of generality that $z_t \in \wideparen{x_t y_t}$.
    Recall that the metric is a circle of diameter $\Delta$, we have:
    %
    \begin{equation}
        \label{eqn:semi-circle-assumption}
        d(x_t, z_t) + d(y_t, z_t) = d(x_t, y_t) < \Delta
        ~.
    \end{equation}
    %
    Moving $y_t$ to $\bar{x}_t$ decreases the distance terms by $2 d(y_t, \bar{x}_t)$, and increases the work functions by at most this amount by Lipschitzness:
    \begin{align*}
        \Phi_t
        &
        = \Phi(w_t, u_t, x_t, y_t, z_t) \\
        & 
        = \Phi^\star(w_t, u_t, x_t, y_t, z_t) - d(x_t, y_t) -d(y_t, z_t) - d(z_t, x_t)
        \tag{Equations~\eqref{eqn:potential}, \eqref{eqn:potential-alternative}} \\
        &
        = \Phi^\star(w_t, u_t, x_t, y_t, z_t) - 2 d(x_t, y_t)
        \tag{Eqn.~\eqref{eqn:semi-circle-assumption}} \\
        &
        = \Phi^\star(w_t, u_t, x_t, y_t, z_t) + d(x_t, \bar{y}_t) + d(\bar{x}_t, y_t) - 2\Delta
        \tag{Eqn.~\eqref{eqn:antipode}} \\
        & 
        \ge \Phi^\star(w_t, u_t, x_t, \bar{x}_t, z_t) - 2\Delta
        \tag{Lipschizness in Lemma~\ref{lem:work-function}} \\
        &
        = \Phi(w_t, u_t, x_t, \bar{x}_t, z_t) \\
        &
        \ge \Phi_t
        ~.
    \end{align*}

    Since both ends are the same, this sequence of inequalities must hold with equality.
    In particular, since $(x_t, \bar{x}_t, z_t) \in P$, we have:
    \[
        \Phi_t = \Phi^\star(w_t, u_t, x_t, \bar{x}_t, z_t) - 2\Delta
        \ge \min_{u \in \pointset, (x, y, z) \in P} \Phi^\star(w_t, u, x, y, z) -2\Delta
    \]
    as needed.
\end{proof}

\begin{proof}[Proof Sketch of Lemma~\ref{lem:3-server}]
    By Lemma~\ref{lem:potential-equivalence}, it suffices to prove that for any $u \in \pointset$ and any $(x, y, z) \in P$, there are $x', y', z' \in P$ such that:
    \begin{equation}
        \label{eqn:3-server-sketch}
        \Phi^\star(u, x, y, z) \ge \Phi^\star(r_t, x', y', z')
        ~.
    \end{equation}

    Specifically, our argument will consider $x', y', z' \in \{u, \bar{u}, x, \bar{x}, y, \bar{y}, z, \bar{z} \}$.
    
    Following the literature, we say that $w_t(X)$ resolves from $x \in X$ for some configuration $X \in \pointset^k$ if in the cost-minimizing solution the server that serves the last request $r_t$ is at $x \in X$ in the end.

    Observe that both sides are the sum of the work function $w_t$'s values on four configurations.
    We will prove it through a sequence of transformations of the following kinds:
    
    \begin{enumerate}
        \item \textbf{Resolving:~}
            Consider all possible points that the configurations on the left resolve from.
            For each configuration $X$, if $w_t(X)$ resolves from $x \in X$, write $w_t(X) = w_t(X-x+r_t)+d(x, r_t)$, where $X-x+r_t$ denotes the configuration obtained by removing $x$ from and adding $r_t$ to $X$.
            We shall reorder the points in each configuration so that $r_t$ is the first after the resolving step.
            \medskip
        \item \textbf{Quasi-convexity:~}
            Pick two configurations $X, Y$ on the left and apply quasi-convexity to swap a pair of points in $X, Y$.
            Importantly, we will always apply the resolving steps first, so that $r_t \in X, Y$ which will not be swapped according to Lemma~\ref{lem:work-function}.
            \medskip
        \item \textbf{Lipschitzness:~}
            For each $d(x, y)$ on the left, find a configuration $X$ containing one of the points, say, $x$, and write $w_t(X) + d(x, y) \ge w_t(X-x+y)$.
    \end{enumerate}

    The full proof unfortunately involves a tedious case analysis, which we defer to Appendix~\ref{app:proof-3-server}.
    Below we present some representative cases to demonstrate the transformations.
    The case numbers are not consecutive because they are cherry-picked from the full analysis.
    For ease of notations we omit the subscript $t$ and write $w = w_t$ and $r = r_t$ below.

    \subparagraph{Case 0:}
    If $w(ux\bar{y}), w(uy\bar{z}), w(uz\bar{x})$ all resolve from $u$, we only need the resolving and Lipschitzness steps, effectively moving all copies of $u$ to $r$.
    \begin{align*}
        \Phi^\star(w, u, x, y, z)
        &
        = w(\bar{u}^3) + w(rx\bar{y}) + w(ry\bar{z}) + w(rz\bar{x}) + 3d(r, u) \tag{Resolving} \\
        &
        \ge w(\bar{r}^3) + w(rx\bar{y}) + w(ry\bar{z}) + w(rz\bar{x}) \tag{Lipschizness} \\[.5ex]
        &
        = \Phi^{\star}(w, r, x, y, z)
        ~.
    \end{align*}
    
    \subparagraph{Subcase 1a:}
    Suppose that exactly two of $w(ux\bar{y}), w(uy\bar{z}), w(uz\bar{x})$ resolve from points other than $u$ (Case 1), and further suppose that these two resolve from points other than the pair of antipodes in their final configurations (Subcase 1a).
    By symmetry, assume without loss of generality that $w(ux\bar{y})$ resolves from $u$.
    The pair of antipodes in the other two are $z$ and $\bar{z}$, and therefore $w(uy\bar{z})$ resolves from $y$ and $w(uz\bar{x})$ resolves from $\bar{x}$.

    We once again only need the resolving and Lipschitzness steps, since moving $y$ and $\bar{x}$ to $r$ allows us to move both $\bar{y}$ and $x$ in the second configuration to $\bar{r}$.
    Finally, the first configuration must resolve from a copy of $\bar{u}$, moving which to $r$ allows us to move $u$ to $\bar{r}$ in the second configuration.
    Formally:
    \begin{align*}
        \Phi^{\star}(w, u, x, y, z)
        &
        = w(r\bar{u}^2) + w(ux\bar{y}) + w(ru\bar{z}) + w(ruz) \\
        & \qquad 
        + d(r, \bar{u}) + d(r, y) + d(r, \bar{x}) \tag{Resolving} \\
        &
        \ge w(r\bar{u}^2) + w(\bar{r}^3) + w(ru\bar{z}) + w(ruz) \tag{Lipschizness} \\[.5ex]
        &
        = \Phi^{\star}(w, r, u, z, \bar{u})
    \end{align*}
    

    \subparagraph{Subcase 2c:}
    Suppose that exactly one of $w(ux\bar{y}), w(uy\bar{z}), w(uz\bar{x})$ resolve from points other than $u$.
    By symmetry, assume without loss of generality that $w(ux\bar{y})$ and $w(uy\bar{z})$ resolve from $u$, and $uz\bar{x}$ resolves from $z$ (Case 2).
    Further suppose that $r \in \wideparen{\bar{y}\bar{z}}$~~(Subcase 2c).
    Then the resolving step is as follows, with an rearrangement of the distance terms:
    \begin{align*}
        \Phi^{\star}(w, u, x, y, z)
        &
        = w(r\bar{u}^2) + w(rx\bar{y}) + w(ry\bar{z}) + w(ru\bar{x}) + d(r, \bar{u}) \\
        & \qquad
        + 2 d(r, u) + d(r, z) \tag{Resolving} \\
        &
        = w(r\bar{u}^2) + w(rx\bar{y}) + w(ry\bar{z}) + w(ru\bar{x}) \\
        & \qquad
        + \Delta + d(r, u) + d(r, z)
        ~.
        \tag{Eqn.~\eqref{eqn:antipode}} 
    \end{align*}

    Next, we will pick $x' \in \{x, \bar{x}\}$, $y' \in \{y, \bar{y} \}$, and $z' \in \{u, \bar{u}\}$ to prove Eqn.~\eqref{eqn:3-server-sketch}.
    We remark that the points $x, y, \bar{x}, \bar{y}$ partition the circle into four arcs.
    Depending on which of the four arcs contain point $u$, there two exactly two choices of $x', y', z'$ in each case subject to $(x', y', z') \in P$.
    For example, if $u \in \wideparen{xy}$ then $(\bar{x}, \bar{y}, u), (x, y, \bar{u}) \in P$.
    The following analysis therefore consists of four subcases.

    If $u \in \wideparen{\bar{x}\bar{y}}$~, consider $(x', y', z') = (x, y, u) \in P$.
    Then, we need to get configurations $\bar{r}^3$, $rx\bar{y}$, $ry\bar{u}$, and $ru\bar{x}$ from the transformations.
    The equation above already has configurations $rx\bar{y}$ and $ru\bar{x}$.
    Further applying quasi-convexity to swap a copy of $\bar{u}$ in $r\bar{u}^2$ with $y$ in $ry\bar{z}$ gives configuration $ry\bar{u}$ (and configuration $r\bar{z}\bar{u}$).
    Finally, Lipschitzness transforms configuration $r\bar{z}\bar{u}$ into $\bar{r}^3$.
    For clarity, we will write the two points being swapped in each quasi-convexity step in bold hereafter.
    Formally:
    \begin{align*}
        \Phi^{\star}(w, u, x, y, z)
        &
        \ge w(r\bm{y}\bar{u}) + w(rx\bar{y}) + w(r\bm{\bar{u}}\bar{z}) + w(ru\bar{x}) \\
        & \qquad
        + \Delta + d(r, u) + d(r, z)
        \tag{Quasi-convexity} \\
        &
        \ge w(ry\bar{u}) + w(rx\bar{y}) + w(\bar{r}^3) + w(ru\bar{x})
        \tag{Lipschitzness} \\[.5ex]
        &
        = \Phi^{\star}(w, r, x, y, u)
        ~.
    \end{align*}

    %

    The other three subcases are the crux of the analysis where we use the assumption $r \in \wideparen{\bar{y}\bar{z}}$~~ to conclude that $\bar{y} \in \wideparen{rz}$.
    This implies $d(r, z) = d(r, \bar{y}) + d(\bar{y}, z)$, which allows us to eliminate point $\bar{z}$ that is not in the final configurations:
    \begin{align*}
        \Phi^{\star}(w, u, x, y, z)
        &
        = w(r\bar{u}^2) + w(rx\bar{y}) + w(ry\bar{z}) + w(ru\bar{x}) + \Delta + d(r, u) + d(r, \bar{y}) + d(\bar{y}, z) \\
        &
        \ge w(r\bar{u}^2) + w(rx\bar{y}) + w(ry^2) + w(ru\bar{x}) + \Delta + d(r, u) + d(r, \bar{y})
        ~.
        \tag{Lipschizness}
    \end{align*}

    If $u \in \wideparen{xy}$, consider $(x', y', z') = (x, y, \bar{u}) \in P$.
    Then, we need to get configurations $\bar{r}^3$, $rx\bar{y}$, $ryu$, and $r\bar{u}\bar{x}$ from the transformations.
    We already have configuration $rx\bar{y}$.
    The remaining argument first applies quasi-convexity to $ry^2$ and $ru\bar{x}$ to obtain configuration $ryu$ in the final formula, and configuration $ry\bar{x}$.
    Further apply quasi-convexity to $r\bar{u}^2$ and $ry\bar{x}$ to obtain configuration $r\bar{u}\bar{x}$ in the final formula.
    Finally, Lipschitzness transforms the last configuration into $\bar{r}^3$.
    Formally:
    \begin{align*}
        \Phi^{\star}(w, u, x, y, z)
        &
        \ge w(r\bar{u}^2) + w(rx\bar{y}) + w(r\bm{\bar{x}}y) + w(ru\bm{y}) \\
        & \qquad
        + \Delta + d(r, u) + d(r, \bar{y})
        \tag{Quasi-convexity} \\[.5ex]
        &
        \ge w(r\bm{y}\bar{u}) + w(rx\bar{y}) + w(r\bar{x}\bm{\bar{u}}) + w(ruy) \\
        & \qquad
        + \Delta + d(r, u) + d(r, \bar{y})
        \tag{Quasi-convexity} \\[.5ex]
        &
        \ge w(\bar{r}^3) + w(rx\bar{y}) + w(r\bar{x}\bar{u}) + w(ruy)
        \tag{Lipschitzness} \\[1ex]
        &
        = \Phi^{\star}(w, r, x, y, \bar{u})
        ~.
    \end{align*}

    If $u \in \wideparen{y\bar{x}}$, consider $(x, \bar{y}, u) \in P$.
    Similar to the previous subcase, comparing the existing configurations with those in the final formula gives a sequence of transformations below:
    \begin{align*}
        \Phi^{\star}(w, u, x, y, z)
        &
        \ge w(r\bar{u}^2) + w(rx\bm{y}) + w(ry\bm{\bar{y}}) + w(ru\bar{x})\\
        & \qquad
        + \Delta + d(r, u) + d(r, \bar{y})
        \tag{Quasi-convexity} \\[.5ex]
        &
        \ge w(r\bar{u}\bm{y}) + w(rxy) + w(r\bm{\bar{u}}\bar{y}) + w(ru\bar{x}) \\
        & \qquad
        + \Delta + d(r, u) + d(r, \bar{y})
        \tag{Quasi-convexity} \\[.5ex]
        &
        \ge w(\bar{r}^3) + w(rxy) + w(r\bar{u}\bar{y}) + w(ru\bar{x}) \tag{Lipschizness} \\[1ex]
        &
        = \Phi(w, r, x, \bar{y}, u)
        ~.
    \end{align*}

    If $u \in \wideparen{x\bar{y}}$, 
    we shall choose the final configurations depending on the result of applying quasi-convexity to $rx\bar{y}$ and $ru\bar{x}$.
    If $w(rx\bar{y}) + w(ru\bar{x}) \ge w(rx\bm{u}) + w(r\bm{\bar{y}}\bar{x})$, consider $(x', y', z') = (\bar{x}, y, u) \in P$ because the right-hand-side already gives two of the final configurations:
    \begin{align*}
        \Phi^{\star}(w, u, x, y, z)
        &
        \ge w(r\bar{u}^2) + w(rx\bm{u}) + w(ry^2) + w(r\bm{\bar{y}}\bar{x}) \\
        & \qquad
        + \Delta + d(r, u) + d(r, \bar{y})
        \tag{Quasi-convexity} \\[.5ex]
        &
        \ge w(r\bar{u}\bm{y}) + w(rxu) + w(r\bm{\bar{u}}y) + w(r\bar{y}\bar{x}) \\
        & \qquad
        + \Delta + d(r, u) + d(r, \bar{y}) \tag{Quasi-convexity} \\[.5ex]
        &
        \ge w(\bar{r}^3) + w(rxu) + w(r\bar{u}y) + w(r\bar{y}\bar{x}) \tag{Lipschizness} \\[1ex]
        &
        = \Phi^{\star}(w, r, \bar{x}, y, u)
        ~.
    \end{align*}

    Otherwise, i.e., if $w(rx\bar{y}) + w(ru\bar{x}) \geq w_t(rx\bm{\bar{x}}) + w(ru\bm{\bar{y}})$, consider $(x', y', z') = (x, \bar{y}, \bar{u}) \in P$ because the right-hand-side gives a configuration $ru\bar{y}$ in the final formula.
    Formally:
    \begin{align*}
        \Phi^{\star}(w, u, x, y, z)
        &
        \ge w(r\bar{u}^2) + w(rx\bm{\bar{x}}) + w(ry^2) + w(ru\bm{\bar{y}}) \\
        & \qquad
        + \Delta + d(r, u) + d(r, \bar{y})
        \tag{Quasi-convexity} \\[.5ex]
        &
        \ge w(r\bar{u}^2) + w(rx\bm{y}) + w(r\bm{\bar{x}}y) + w(ru\bar{y}) \\
        & \qquad
        + \Delta + d(r, u) + d(r, \bar{y}) \tag{Quasi-convexity} \\[.5ex]
        &
        \ge w(r\bar{u}\bm{y}) + w(rxy) + w(r\bar{x}\bm{\bar{u}}) + w(ru\bar{y}) \\
        & \qquad
        + \Delta + d(r, u) + d(r, \bar{y}) \tag{Quasi-convexity} \\[.5ex]
        &
        \ge w(\bar{r}^3) + w(rxy) + w(r\bar{x}\bar{u}) + w(ru\bar{y}) \tag{Lipschizness} \\[1ex]
        &
        = \Phi^{\star}(w, x, \bar{y}, \bar{u})
        ~.
    \end{align*}
\end{proof}

\section{Canonical Potentials}
\label{sec:potential}

This section considers a canonical class of potential functions which capture as special cases many existing potential functions in the literature, including ours in Section~\ref{sec:3-server}.
We shall examine whether these canonical potentials could resolve the deterministic $k$-server conjecture on \emph{general metrics}.
Since we consider general metrics, we will assume without loss of generality that for any point $x \in \pointset$ there is an antipode $\bar{x} \in \pointset$ such that for any $y \in \pointset$, $d(x, y) + d(y, \bar{x}) = \Delta$, where $\Delta > 0$ is the diameter of the metric.

A canonical potential function is parameterized by a subset $P$ of distinct pairs of points $(i,j) \ne (i',j')$, where $i, i' \in [k]$ and $j, j' \in [k-1]$.
It takes as input a work function $w_t$, a point $u$, and an ordered set of points $X = (x_{ij})_{i \in [k], j \in [k-1]}$.
For any $i \in [k]$, we consider $X_i = \{ x_{ij} \}_{j \in [k-1]} \in \pointset^{k-1}$ as a configuration of $k-1$ points, and $uX_i$ as a configuration of $k$ points.
The canonical potential is: 
\begin{equation}
    \label{eqn:canonical-potential}
    \Phi(w_t, u, X) = w_t(\bar{u}^k) + \sum_{i=1}^{k} w_t(uX_i) - \sum_{\{(i,j), (i',j')\} \in P} d(x_{i,j}, x_{i',j'})
    ~.
\end{equation}

The distance terms can enforce various combinatorial constraints on the points in $X$.
For example, letting a point pair $(i,j), (i',j')$ be in $P$ may impose a constraint that the two points are antipodes, i.e., $x_{i,j} = \bar{x}_{i',j'}$.
Further, the potential function in Section~\ref{sec:3-server} uses the distance terms to ensure that three points are not on the same semi-circle in the case of circle metrics (Lemma~\ref{lem:potential-equivalence}).


Then, let the potential value at any time $0 \le t \le T$ be:
\begin{equation}
    \label{eqn:canonical-potential-value}
    \Phi_t = \min_{u, X} \Phi(w_t, u, X)
    ~.
\end{equation}

Although one may further generalize the canonical potentials, such as considering arbitrary coefficients of the distance terms instead of only $0$ and $-1$, all examples for $k$-competitiveness to our knowledge only need the above more restricted form. \footnote{The potential proposed by \citet{KoutsoupiasP:JACM:1995} is used to show $(2k-1)$-competitiveness in general metric, which is not captured by our canonical form, but can be captured by a generalized version. See Appendix B for details.}
Below are some example potential functions in the literature, and we state without proofs how to write them in the canonical form (up to an fixed additive term). 
Appendix~\ref{app:potential} includes the detailed arguments for completeness. 

\begin{example}[This Paper, Circle Metrics, $k = 3$]
    \label{exa:this-paper}
    Our potential function in Eqn.~\eqref{eqn:potential} can be equivalently written in the canonical form as:
    \begin{align*}
        \Phi_\text{\rm Canonical} (w_t, u, X)
        &
        = w_t(\bar{u}^3) + w_t(uX_1) + w_t(uX_2) + w_t(uX_3) \\
        & \quad
        - d(x_{12}, x_{21}) - d(x_{22}, x_{31}) - d(x_{32}, x_{11}) \tag{Antipode constraints} \\
        & \quad
        - d(x_{11}, x_{21}) - d(x_{21}, x_{31}) - d(x_{31}, x_{11}) \tag{Not on a semi-circle}
        ~.
    \end{align*}
\end{example}

\medskip

\begin{example}[\citet{chrobak1991server}, General Metrics, $k=2$]
    \label{exa:chrobak-larmore}
    For any work function $w_t$ from a $2$-server instance with $r_t$ as the last request, their original potential function is:
    \[
        \Phi^\text{\rm CL} (w_t) = \min_{x,y,z \in \pointset} w_t(r_t x) + w_t(r_t y) + w_t(r_t z) - d(r_t, x) - d(y, z)
        ~,
    \]
    Since work functions are Lipchitzness, and because each of $x$ and $z$ only appears in one configuration, it is equivalent to
    \begin{align*}
        \Phi^\text{\rm CL} (w_t) & = \min_{y \in \pointset} w_t(r_t \bar{r}_t) + w_t(r_t y) + w_t(r_t \bar{y}) - 2\Delta \\
        & = w_t(\bar{r}^2_t) + w_t(r_t y) + w_t(r_t \bar{y}) - \Delta
        ~.
    \end{align*}
    Finally, even if we replace $r_t$ with an unconstrained point $u$, the potential would still achieve its minimum value at $u = r_t$.
    We can therefore write it in the canonical form as:
    \[
        \Phi^\text{\rm CL}_\text{\rm Canonical}(w_t, u, X) = w_t(\bar{u}^2) + w_t(ux_{11}) + w_t(ux_{21}) - d(x_{11}, x_{21})
        ~.
    \]
\end{example}

\medskip

\begin{example}[\citet*{BeinCL:TCS:2002}, Manhattan plane, $k=3$]
    \label{exa:BeinCL}
    For any work function $w_t$ from a $3$-server instance on a Manhattan plane with $r_t$ as the last request, their original potential is defined as:
    \begin{align*}
        \Phi^\text{\rm BCL} (w_t)
        = \min_{a,a',b,b',p \in \pointset, C \in \pointset^3} & ~
        w_t(C) + w_t(r_t pb) + w_t(r_t pb') + w_t(r_t aa') \\
        & \quad
        - d(C, r_t^3) - d(b, b') - d(p, a) - d(p, a')
        ~.
    \end{align*}
    Recall that we focus on general metrics that admit antipodes.
    The distance term $d(C, r_t^3)$ effectively enforces that $C = \bar{r}_t^3$, i.e.:
    \begin{align*}
        \Phi^\text{\rm BCL} (w_t)
        = \min_{a,a',b,b',p \in \pointset} & ~
        w_t(\bar{r}_t^3) + w_t(r_t pb) + w_t(r_t pb') + w_t(r_t aa') \\
        & \quad
        - d(b, b') - d(p, a) - d(p, a') - 3\Delta
        ~.
    \end{align*}
    Further, since work functions are Lipschitz and because each of $a$ and $a'$ appears in only one configuration.
    The above is equivalent to:
    \begin{align*}
        \Phi^\text{\rm BCL} (w_t)
        = \min_{b,b',p \in \pointset} & ~
        w_t(\bar{r}_t^3) + w_t(r_t pb) + w_t(r_t pb') + w_t(r_t \bar{p}^2)
        - d(b, b') - 5\Delta
        ~.
    \end{align*}
    Finally, even if we replace $r_t$ with an unconstrained point $u$, the potential would still achieve its minimum value at $u = r_t$.
    We can therefore write it in the canonical form as:
    \begin{align*}
        \Phi^\text{\rm BCL}_\text{\rm Canonical} (w_t, u, X)
        &
        = w_t(\bar{u}^3) + w_t(uX_1) + w_t(uX_2) + w_t(uX_3) - d(\overbrace{x_{12}}^{b}, \overbrace{x_{22}}^{b'}) \\
        & \qquad 
        - d(x_{11}, x_{31}) - d(x_{11}, x_{32})
        - d(x_{21}, x_{31}) - d(x_{21}, x_{32})
        ~.
    \end{align*}
\end{example}

\medskip

\begin{example}[\citet{CoesterK:ICALP:2021}, Various Metrics, e.g., Multi-ray Metrics, General Metrics with $k=2$, etc.]
    \label{exa:Coester-Koutsoupias}
    Their original potential function is:
    %
    \[
        \Phi^\text{\rm CK} (w, x_1, x_2, \dots, x_k) = \sum_{i = 1}^{k+1} w \big( \bar{x}_{i-1}^{i-1} x_i x_{i+1}\dots x_k \big)
        ~.
    \]
    To write it in the canonical form, we use the distance terms to impose appropriate antipode constraints, which implicitly enforce the identity constraints as well:
    \[
        \Phi^\text{\rm CK}_\text{\rm Canonical} (w_t, u, X) = w_t(\bar{u}^k) + \sum_{i=1}^k w_t(uX_i) - \sum_{i=1}^k \sum_{j=i+1}^k \sum_{\ell=i}^{k-1} d(x_{i\ell}, x_{ji})
        ~.
    \]
\end{example}

\subsection{Canonical Argument}

Recall that we typically aim to prove Equations~\eqref{eqn:potential-approach-stepwise} and \eqref{eqn:potential-approach-bounded}, which we restate below:
\begin{align*}
    \Phi_t - \Phi_{t-1} & \ge \max_{X \in \pointset^k} \big( w_t(X) - w_{t-1}(X) \big) ~, && \forall 1 \le t \le T; \\
    \Phi_T - \Phi_0 & \le (k+1) \OPT + O(1) ~.
\end{align*}

Further, the canonical potentials satisfy the second inequality by definition.
It remains to argue that the change in the potential values from step $t-1$ to step $t$ upper bounds the extended cost on the right-hand-side.
To do so, we only need to prove that the potential achieves its minimum value when $u$ equals the last request (see, e.g., the proof of Theorem~\ref{thm:3-server} given Lemma~\ref{lem:3-server}).
In other words, if for any work function $w$ from a $k$-server instance with $r$ as the last request:
\begin{equation}
    \label{eqn:canonical-argument}
    \exists X^* \in \pointset^{k(k-1)} ~: \qquad (r, X^*) \in \arg\min_{u, X} ~ \Phi(w, u, X)
    ~.
\end{equation}

In fact, existing arguments do not inspect the request history behind the work function.
All they need is that the work function is induced from a collection of support sets and their corresponding values, where the last request $r_t$ belongs to every support set.
More precisely, suppose that $(S_i, v_i)$, $1 \le i \le m$, are a collection of support-set-value pairs, where $S_i \in \pointset^k$ and $v_i \in \R$.
Then, the corresponding work function $w$ is:
\begin{equation}
    \label{eqn:support-set-induce-function}
    w(X) = \min_{1 \le i \le m} \Big( v_i + d(S_i, X) \Big)
    ~.
\end{equation}

This leads to a natural question:
\emph{Is there a canonical potential that satisfies Eqn.~\eqref{eqn:canonical-argument} for any function $w$ defined by a collection of support-set-value pairs and the above Eqn.~\eqref{eqn:support-set-induce-function}, and any point $r$ that belongs to every support set?}

\subsection{Limitation of Canonical Potentials}




Unfortunately, the answer is negative even for $k=3$ servers, via a computer-aided argument.
Recall that the case of $k=2$ has already been resolved by Chrobak and Larmore~\cite{chrobak1991server} using a potential that can be rewritten in the canonical form.
This negative result shows an intriguing separation between the $2$-server and $3$-server problems.

\begin{theorem}
    \label{thm:canonical-limit}
    No canonical potential (which is defined for competitive ratio $\Gamma = k$) satisfies the canonical argument Eqn.~\eqref{eqn:canonical-argument}  for any function $w$ defined by a collection of support-set-value pairs and the above Eqn.~\eqref{eqn:support-set-induce-function}, and any point $r$ that belongs to every support set, even when $k = 3$.
\end{theorem}

\begin{proof}
    Since $k=3$, there are only $6$ points $x_{ij}$ where $1 \le i \le 3$ and $1 \le j \le 2$.
    Hence, there are only $\binom{6}{2} = 15$ pairs of distinct points, and only $2^{15}$ different canonical potentials.
    Using a computer program,%
    \footnote{See \url{https://github.com/FoolMike/K-server-Canonical-Potential} for the code.}
    we show that each of them fails to satisfy Eqn.~\eqref{eqn:canonical-argument} on at least one of the following three test cases.
    
    \subparagraph*{Test Case (a):}
    This is the counter example from Coester and Koutsoupias~\cite{CoesterK:ICALP:2021}, who used it to refute the lazy adversary conjecture.
    Consider a circle metric with $8$ points names $0$ to $7$ positioned clockwise, such that the distance between each pair of neighboring points equals $1$.
    Consider $r = 4$ and let the support sets and values be $w(456) = w(457) = 8$, $w(436) = w(437) = 9$, $w(416) = w(417) = w(425) = 10$, $w(423) = 11$.
    See Figure~\ref{fig:test-case-a} for an illustration.
    
    \subparagraph*{Test Case (b):}
    The second test case considers the same circle metric, but a different collection of support sets and values.
    Still let $r = 4$ be the last request and let the support sets and values be $w(456) = 8$, $w(405) = 9$, $w(416) = w(425) = w(446) = 10$, $w(401) = w(402) = 11$, $w(422) = 12$.
    See Figure~\ref{fig:test-case-b} for an illustration.
    
    \subparagraph*{Test Case (c):}
    The last test case is a three-dimensional Manhattan cube. 
    Name the remaining points as $0$ to $7$. Let the last request be $7$ and let the support sets and values be $w(701) = w(702) = w(704) = 0$.
    See Figure~\ref{fig:test-case-c} for an illustration.
\end{proof}

In fact, the computer verification further shows that our potential is the only canonical potential that satisfies Eqn.~\eqref{eqn:canonical-argument} for $k = 3$ and the circle metrics, noting that the following potential is equivalent to ours:
\begin{align*}
    \Phi(w_t, u, x, y, z)
    &
    = w_t(\bar{u}^3) + w_t(ux\bar{y}) + w_t(uy\bar{z}) + w_t(uz\bar{x}) \\
    & \qquad
    - d(x, y) - d(y, z) - d(z, x) - d(\bar{x}, \bar{y}) - d(\bar{y}, \bar{z}) - d(\bar{z}, \bar{x})
    ~.
\end{align*}

\tikzset{point/.style={circle,fill=black,scale=.5}}
\tikzset{metric/.style={black,thick}}
\tikzset{request/.style={star,fill=red,scale=.7}}
\tikzset{support/.style={dashed,red,thick}}

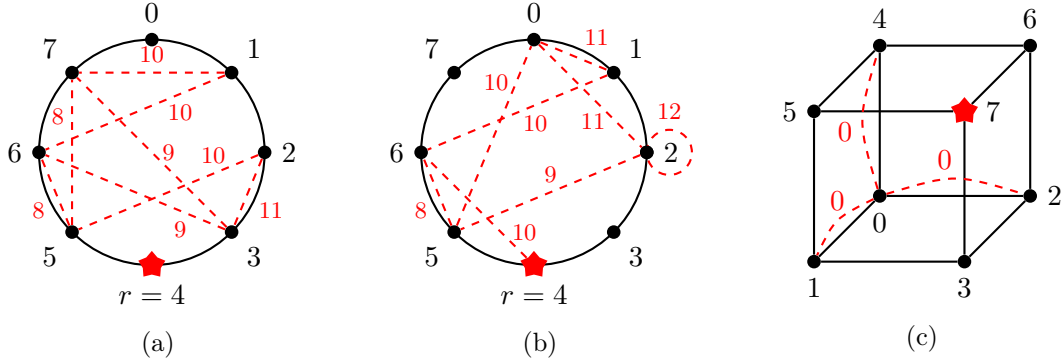
\begin{figure}[t]
    \centering
    \begin{subfigure}{.3\textwidth}
        \centering
        \begin{tikzpicture}
            \draw[white,dotted] (-2,-2.1) rectangle +(4.2,4.2);
            \filldraw[draw=black,thick,fill=none] (0,0) circle (1.5);
            \node[point,label=above:$0$] (0) at (0,1.5) {};
            \node[point,label=above right:$1$] (1) at ({1.5/sqrt(2)},{1.5/sqrt(2)}) {};
            \node[point,label=right:$2$] (2) at (1.5,0) {};
            \node[point,label=below right:$3$] (3) at ({1.5/sqrt(2)},-{1.5/sqrt(2)}) {};
            \node[request,label=below:{$r=4$}] (4) at (0,-1.5) {};
            \node[point,label=below left:$5$] (5) at (-{1.5/sqrt(2)},-{1.5/sqrt(2)}) {};
            \node[point,label=left:$6$] (6) at (-1.5,0) {};
            \node[point,label=above left:$7$] (7) at ({-1.5/sqrt(2)},{1.5/sqrt(2)}) {};
            \draw[support] (5)--(6) node[below left,midway,red] {\footnotesize $8$};
            \draw[support] (5)--(7) node[left=-1pt,pos=.75,red] {\footnotesize $8$};
            \draw[support] (3)--(6) node[below,near start,red] {\footnotesize $9$};
            \draw[support] (3)--(7) node[right,midway,red] {\footnotesize $9$};
            \draw[support] (1)--(6) node[below,near start,red] {\footnotesize $10$};
            \draw[support] (1)--(7) node[above,midway,red] {\footnotesize $10$};
            \draw[support] (2)--(5) node[above,near start,red] {\footnotesize $10$};
            \draw[support] (2)--(3) node[below right,midway,red] {\footnotesize $11$};
        \end{tikzpicture}
        \caption{}
        \label{fig:test-case-a}
    \end{subfigure}
    \begin{subfigure}{.3\textwidth}
        \centering
        \begin{tikzpicture}
            \draw[white,dotted] (-2,-2.1) rectangle +(4.2,4.2);
            \filldraw[draw=black,thick,fill=none] (0,0) circle (1.5);
            \node[point,label=above:$0$] (0) at (0,1.5) {};
            \node[point,label=above right:$1$] (1) at ({1.5/sqrt(2)},{1.5/sqrt(2)}) {};
            \node[point,label=right:$2$] (2) at (1.5,0) {};
            \node[point,label=below right:$3$] (3) at ({1.5/sqrt(2)},-{1.5/sqrt(2)}) {};
            \node[request,label=below:{$r=4$}] (4) at (0,-1.5) {};
            \node[point,label=below left:$5$] (5) at (-{1.5/sqrt(2)},-{1.5/sqrt(2)}) {};
            \node[point,label=left:$6$] (6) at (-1.5,0) {};
            \node[point,label=above left:$7$] (7) at ({-1.5/sqrt(2)},{1.5/sqrt(2)}) {};
            \draw[support] (5)--(6) node[below left,midway,red] {\footnotesize $8$};
            \draw[support] (5)--(0) node[left=-1pt,pos=.8,red] {\footnotesize $10$};
            \draw[support] (1)--(6) node[below,pos=.4,red] {\footnotesize $10$};
            \draw[support] (2)--(5) node[above,midway,red] {\footnotesize $9$};
            \draw[support] (4)--(6) node[right,near start,red] {\footnotesize $10$};
            \draw[support] (0)--(1) node[above right,midway,red] {\footnotesize $11$};
            \draw[support] (0)--(2) node[left,near end,red] {\footnotesize $11$};
            \draw[support] (2.1,0) arc (0:360:.3) node[above,pos=.25,red] {\footnotesize $12$};
            \node[point] at (2) {};
        \end{tikzpicture}
        \caption{}
        \label{fig:test-case-b}
    \end{subfigure}
    \begin{subfigure}{.3\textwidth}
        \centering
        \begin{tikzpicture}
            \node[point,label=below:$0$] (0) at (0,0) {};
            \node[point,label=below:$1$] (1) at ({-.618*sqrt(2)},{-.618*sqrt(2)}) {};
            \node[point,label=right:$2$] (2) at (2,0) {};
        	\node[point,label=below:$3$] (3) at ({2-.618*sqrt(2)},{-.618*sqrt(2)}) {};
        	\node[point,label=above:$4$] (4) at (0,2) {};
        	\node[point,label=left:$5$] (5) at ({-.618*sqrt(2)},{2-.618*sqrt(2)}) {};
        	\node[point,label=above:$6$] (6) at (2,2) {};
        	\node[request,label=right:$7$] (7) at ({2-.618*sqrt(2)},{2-.618*sqrt(2)}) {};
            \draw[metric] (0)--(1);
            \draw[metric] (0)--(2);
            \draw[metric] (0)--(4);
            \draw[metric] (1)--(3);
            \draw[metric] (1)--(5);
            \draw[metric] (2)--(3);
            \draw[metric] (2)--(6);
            \draw[metric] (3)--(7);
            \draw[metric] (4)--(5);
            \draw[metric] (4)--(6);
            \draw[metric] (5)--(7);
            \draw[metric] (6)--(7);
            \draw[support] (0) .. controls (1,0.3) .. (2) node[pos=.4,above,red] {$0$};
            \draw[support] (0) .. controls (-0.3, 1) .. (4) node[pos=.4,left,red] {$0$};
            \draw[support] (0) .. controls (-0.6, -0.3) .. (1) node[midway,above,red] {$0$};
        \end{tikzpicture}
        \caption{}
        \label{fig:test-case-c}
    \end{subfigure}
    \caption{Illustrative figures for the test cases used in Theorem~\ref{thm:canonical-limit}. In each test case, the red star denotes the last request; each dashed red arc $ab$, including the self-loop in the second test case, and its red label $v$ represent a support set $w(rab)=v$. The metric is the shortest path distance w.r.t.\ the solid black arcs.}
    \label{fig:test-case}
\end{figure}

\bibliographystyle{plainnat}
\bibliography{k-server}

\appendix



%

%


\section{Proof of Lemma~\ref{lem:3-server}}
\label{app:proof-3-server}

Recall that $w_t(X)$ resolves from $x$ if $x \in X$ and $w_t(X) = w_t(X - x + r) + d(x, r)$. Let $\Phi_t = \Phi^{\star}(w_t, u, x, y, z) + \Delta$ for some $(x, y, z) \in P$.
In other words, points $x, y, z$ are not in a semi-circle. 

Similar to the proof sketch in the main text, we will drop the subscript $t$ and write $w_t$ as $w$, $r_t$ as $r$ for ease of notations.

\paragraph*{\texorpdfstring{Case 0: all of $w(ux\bar{y}), w(uy\bar{z}), w(uz\bar{x})$ resolve from $u$.}{Case 0}}

This case is already covered in the main body.
We restate it below to be self-contained.
\begin{align*}
    \Phi^{\star}(w, u, x, y, z) & = w(\bar{x}^3) + w(ux\bar{y}) + w(uy\bar{z}) + w(uz\bar{x}) \\
    & = w(\bar{x}^3) + w(ry\bar{z}) + w(rz\bar{w}) + w(rw\bar{y}) + 3d(r, u)
    \tag{Resolving} \\
    & \geq w(\bar{r}^3) + w(ry\bar{z}) + w(rz\bar{w}) + w(rw\bar{y})
    \tag{Lipschitzness} \\
    & = \Phi^{\star}(w, r, x, y, z)
    ~.
\end{align*}

\paragraph*{\texorpdfstring{Case 1: one of $w(ux\bar{y}), w(uy\bar{z}), w(uz\bar{x})$ resolves from $u$.}{Case 1}}

By the symmetry of the potential function, assume without loss of generality that $w(ux\bar{y})$ resolves from $u$. 

\subparagraph{Subcase 1a: $w(uy\bar{z})$ resolves from $y$ and $w(uz\bar{x})$ resolves from $\bar{x}$.}
This subcase is relatively simple as we only need the resolving and Lipschitzness steps, as we have covered in the main text.
In fact, \textbf{we will not use the assumption that $w(ux\bar{y})$ resolves from $u$}.
This feature will help simplify the analysis of Case 3.
\begin{align*}
    \Phi^{\star}(w, u, x, y, z) & = w(\bar{u}^3) + w(ux\bar{y}) + w(uy\bar{z}) + w(uz\bar{x}) \\
    & = w(r\bar{u}^2) + w(ux\bar{y}) + w(ru\bar{z}) + w(ruz) \\
    & \qquad
    + d(r, y) + d(r, \bar{x}) + d(r, \bar{u}) 
    \tag{Resolving} \\
    & \geq w(r\bar{u}^2) + w(\bar{r}^3) + w(ru\bar{z}) + w(ruz)  \tag{Lipschitzness} \\
    & = \Phi^{\star}(w, r, u, z, \bar{u})
    ~.
\end{align*}

\subparagraph*{Subcase 1b: $w(uy\bar{z})$ resolves from $\bar{z}$ and $w(uz\bar{x})$ resolves from $z$.}
First we apply the resolving step and get that:
\begin{align*}
    \Phi^{\star}(w, u, x, y, z) & = w(\bar{u}^3) + w(ux\bar{y}) + w(uy\bar{z}) + w(uz\bar{x}) \\
    & = w(r\bar{u}^2) + w(rx\bar{y}) + w(ruy) + w(ru\bar{x}) + 2\Delta
    ~.
    \tag{Resolving}
\end{align*}

If $u \in \wideparen{xy}$, then $(x, y, \bar{u}) \in P$.
We further have:
\begin{align*}
    \Phi^{\star}(w, u, x, y, z)
    & \geq w(r\bar{u}\bm{u}) + w(rx\bar{y}) + w(ruy) + w(r\bm{\bar{u}}\bar{x}) + 2\Delta 
    \tag{Quasi-convexity} \\
    & \geq w(\bar{r}^3) + w(rx\bar{y}) + w(ruy) + w(r\bar{u}\bar{x}) \tag{Lipschitzness} \\
    & = \Phi^{\star}(w, r, x, y, \bar{u})
    ~.
\end{align*}

If $u \in \wideparen{\bar{x}\bar{y}}$~~, then $(x, y, u) \in P$.
We further have:
\begin{align*}
    \Phi^{\star}(w, u, x, y, z)
    & \geq w(r\bar{u}\bm{u}) + w(rx\bar{y}) + w(r\bm{\bar{u}}y) + w(ru\bar{x}) + 2\Delta
    \tag{Quasi-convexity} \\
    & \geq w(\bar{r}^3) + w(rx\bar{y}) + w(r\bar{u}y) + w(ru\bar{x})
    \tag{Lipschitzness} \\
    & = \Phi^{\star}(w, r, x, y, u)
    ~.
\end{align*}

Next consider the remaining two cases, i.e., if $u \in \wideparen{x\bar{y}}$ of if $u \in \wideparen{y\bar{x}}$.
Apply quasi-convexity to $w(rx\bar{y})$ and $w(ruy)$.
If $w(rx\bar{y}) + w(ruy) \geq w(rx\bm{u}) + w(r\bm{\bar{y}}y)$, then:
\begin{align*}
    \Phi^{\star}(w, u, x, y, z)
    & \geq w(r\bar{u}^2) + w(rx\bm{u}) + w(r\bm{\bar{y}}y) + w(ru\bar{x}) + 2\Delta
    \tag{Quasi-convexity} \\
    & \geq w(r\bar{u}^2) + w(rxu) + w(\bar{r}^3) + w(ru\bar{x})
    \tag{Lipschitzness} \\
    & = \Phi^{\star}(w, r, u, x, \bar{u})
    ~.
\end{align*}

Otherwise, we have:
\begin{align*}
    \Phi^{\star}(w, u, x, y, z)
    & \geq w(r\bar{u}^2) + w(r\bm{u}\bar{y}) + w(r\bm{x}y) + w(ru\bar{x}) + 2\Delta
    ~.
    \tag{Quasi-convexity}
\end{align*}

If $u \in \wideparen{x\bar{y}}$, then $(u, y, \bar{x}) \in P$.
We further have:
\begin{align*}
    \Phi^{\star}(w, u, x, y, z)
    & \geq w(r\bar{u}\bm{u}) + w(ru\bar{y}) + w(rxy) + w(r\bm{\bar}{u}\bar{x}) + 2\Delta
    \tag{Quasi-convexity} \\
    & \geq w(\bar{r}^3) + w(ru\bar{y}) + w(rxy) + w(r\bar{u}\bar{x})
    \tag{Lipschitzness} \\
    & = \Phi^{\star}(w, r, u, y, \bar{x})
    ~.
\end{align*}

If $u \in \wideparen{y\bar{x}}$, then $(u, x, \bar{y}) \in P$.
We further have:
\begin{align*}
    \Phi^{\star}(w, u, x, y, z)
    & \geq w(r\bar{u}\bm{u}) + w(r\bm{\bar{u}}\bar{y}) + w(rxy) + w(ru\bar{x}) + 2\Delta
    \tag{Quasi-convexity} \\
    & \geq w(\bar{r}^3) + w(r\bar{u}\bar{y}) + w(rxy) + w(ru\bar{x})
    \tag{Lipschitzness} \\
    & = \Phi^{\star}(w, r, u, x, \bar{y})
    ~.
\end{align*}

\subparagraph{Subcase 1c: $w(uy\bar{z})$ resolves from $y$ and $w(uz\bar{x})$ resolves from $z$.}
Once we prove this subcase, the remaining subcase when $w(uy\bar{z})$ resolves from $\bar{z}$ and $w(uz\bar{x})$ resolves from $\bar{x}$ follows by symmetry.
Here we need to consider the location of $r$. 

If $r \in \wideparen{\bar{x}\bar{z}}$~~, then by $\bar{x} \in \wideparen{rz}$ we get that $w(uz\bar{x})$ can also resolve from $\bar{x}$.
Hence, it reduces to Subcase 1a. 

If $r \in \wideparen{\bar{y}\bar{z}}$~~, then by $\bar{z} \in \wideparen{ry}$ we get that $w(uy\bar{z})$ can also resolve from $\bar{z}$.
Hence, it reduces to Subcase 1b. 

Finally, consider $r \in \wideparen{\bar{x}\bar{y}}$~~.
We have that $(x, y, r) \in P$ and therefore $\Phi^{\star}(w, u, x, y, z) \geq \Phi^{\star}(w, u, x, y, r)$.
As a result, we may assume without loss of generality that $z = r$.
Next we move $y$ towards $r$, during which $\Phi^{\star}(w, u, x, y, r)$ will not increase.
We may move $y$ all the way to $\bar{x}$ while maintaining $(x, y, r) \in P$.
Hence, we may further assume without loss of generality that $y = \bar{x}$.
We then get that:
\begin{align*}
    \Phi^{\star}(w, u, x, \bar{x}, r) & = w(\bar{u}^3) + w(ux^2) + w(u\bar{x}\bar{r}) + w(ur\bar{x}) \\
    & = w(r\bar{u}^2) + w(rx^2) + w(ru\bar{r}) + w(ru\bar{x}) \\
    & \qquad
    + \Delta + d(r, \bar{x})
    \tag{Resolving} \\
    & \geq w(r\bar{u}^2) + w(rx\bm{u}) + w(r\bm{x}\bar{r}) + w(ru\bar{x}) \\
    & \qquad
    + \Delta + d(r, \bar{x})
    \tag{Quasi-convexity} \\
    & \geq w(r\bar{u}^2) + w(rxu) + w(\bar{r}^3) + w(ur\bar{x})
    \tag{Lipschitzness} \\
    & = \Phi^{\star}(w, r, u, x, \bar{u})
    ~.
\end{align*}

\paragraph*{\texorpdfstring{Case 2: two of $w(ux\bar{y}), w(uy\bar{z}), w(uz\bar{x})$ resolve from $u$.}{Case 2}}
Since the potential function is invariant to changing points $x, y, z$ to their antipodes \emph{and} changing the parity of their permutation, we may assume without loss of generality that the only configuration that does not resolve from $u$ must resolve from $x$, $y$, or $z$.
Further, points $\bar{x}, \bar{y}, \bar{z}$ partition the circle into three arcs.
By that the potential function and the first assumption are both invariant to permuting $x, y, z$ without changing parity, we will assume without loss of generality that $r \in \wideparen{\bar{y}\bar{z}}$.
We then have three subcases depending on which configuration does not resolve from $u$.

\subparagraph*{Subcase 2a: $w(ux\bar{y})$ resolves from $x$.}
By the assumption that $r \in \wideparen{\bar{y}\bar{z}}$~~, we have that $(r, y, z) \in P$.
Further by the assumption of the subcase that $w(ux\bar{y})$ resolves from $x$, we have $\Phi^{\star}(w, u, x, y, z) \geq \Phi^{\star}(w, u, r, y, z)$.
Therefore, we may assume without loss of generality that $x = r$. 
Then we apply the resolving step to get that:
\begin{align*}
    \Phi^{\star}(w, u, r, y, z) & = w(\bar{u}^3) + w(ur\bar{y}) + w(uy\bar{z}) + w(uz\bar{r}) \\
    & = w(r\bar{u}^2) + w(ru\bar{y}) + w(ry\bar{z}) + w(rz\bar{r}) + \Delta + d(r, u)
    ~.
    \tag{Resolving}
\end{align*}

The remaining argument depends on the relative location of $u$ w.r.t.\ $y, \bar{y}, z, \bar{z}$.
If $u \in \wideparen{\bar{y}\bar{z}}$~~, then $(u, y, z) \in P$.
We further have:
\begin{align*}
    \Phi^{\star}(w, u, x, y, z)
    & \geq w(r\bar{u}\bm{\bar{r}}) + w(ru\bar{y}) + w(ry\bar{z}) + w(rz\bm{\bar{u}}) \\
    & \qquad
    + \Delta + d(r, u)
    \tag{Quasi-convexity} \\
    & \geq w(\bar{r}^3) + w(ru\bar{y}) + w(ry\bar{z}) + w(rz\bar{u})
    \tag{Lipschitzness} \\
    & = \Phi^{\star}(w, r, u, y, z)
    ~.
\end{align*}

If $u \in \wideparen{z\bar{y}}$, then $(u ,y, \bar{z}) \in P$.
The argument of this case depends on the result of applying quasi-convexity to $w(ry\bar{z})$ and $w(rz\bar{r})$.
If $w(ry\bar{z}) + w(rz\bar{r}) \geq w(ry\bm{z}) + w(r\bm{\bar{z}}\bar{r})$, then:
\begin{align*}
    \Phi^{\star}(w, u, x, y, z)
    & \geq w(r\bar{u}^2) + w(ru\bar{y}) + w(ry\bm{z}) + w(r\bm{\bar{z}}\bar{r}) \\
    & \qquad 
    + \Delta + d(r, u)
    \tag{Quasi-convexity} \\
    & \geq w(r\bar{u}\bm{\bar{z}}) + w(ru\bar{y}) + w(ryz) + w(r\bm{\bar{u}}\bar{r}) \\
    & \qquad
    + \Delta + d(r, u)
    \tag{Quasi-convexity} \\
    & \geq w(r\bar{u}\bar{z}) + w(ru\bar{y}) + w(ryz) + w(\bar{r}^3)
    \tag{Lipschitzness} \\
    & = \Phi^{\star}(w, r, u, y, \bar{z})
    ~.
\end{align*}

Otherwise, we have:
\begin{align*}
    \Phi^{\star}(w, u, x, y, z)
    & \geq w(r\bar{u}^2) + w(ru\bar{y}) + w(ry\bm{\bar{r}}) + w(rz\bm{\bar{z}}) \\
    & \qquad
    + \Delta + d(r, u)
    ~.
    \tag{Quasi-convexity}
\end{align*}

Further, the assumptions that $u$ lies on arc from $z$ to $\bar{y}$ that $r$ lies on arc from $\bar{y}$ to $\bar{z}$ imply that $\bar{y} \in \wideparen{ru}$.
Therefore:
\begin{align*}
    \Phi^{\star}(w, u, x, y, z)
    & = w(r\bar{u}^2) + w(ru\bar{y}) + w(ry\bar{r}) + w(rz\bar{z}) \\
    & \qquad
    + \Delta + d(r, \bar{y}) + d(\bar{y}, u)
    \tag{$\bar{y} \in \wideparen{ru}$} \\
    & \geq w(r\bar{u}^2) + w(ru^2) + w(\bar{r}^3) + w(rz\bar{z})
    \tag{Lipschitzness} \\
    & \geq w(r\bar{u}^2) + w(ru\bm{\bar{z}}) + w(\bar{r}^3) + w(rz\bm{u})
    \tag{Quasi-convexity} \\
    & = \Phi^{\star}(w, r, u, z, \bar{u})
    ~.
\end{align*}

If $u \in \wideparen{y\bar{z}}$, then $(u, \bar{y}, z) \in P$.
The argument of this case depends on the result of applying quasi-convexity to $w(ru\bar{y})$ and $w_t(ry\bar{z})$.
If $w(ru\bar{y}) + w(ry\bar{z}) \geq w(ru\bm{y}) + w(r\bm{\bar{y}}\bar{z})$, then:
\begin{align*}
    \Phi^{\star}(w, u, x, y, z) 
    & \geq w(r\bar{u}^2) + w(ru\bm{y}) + w(r\bm{\bar{y}}\bar{z}) + w(rz\bar{r}) \\
    & \qquad
    + \Delta + d(r, u) 
    \tag{Quasi-convexity} \\
    & \geq w(r\bm{z}\bar{u}) + w(ruy) + w(r\bar{y}\bar{z}) + w(r\bm{\bar{u}}\bar{r}) \\
    & \qquad
    + \Delta + d(r, u)
    \tag{Quasi-convexity} \\
    & \geq w(rz\bar{u}) + w(ruy) + w(r\bar{y}\bar{z}) + w(\bar{r}^3)
    \tag{Lipschitzness} \\
    & = \Phi^{\star}(w, r, u, \bar{y}, z)
    ~.
\end{align*}

Otherwise, we have $w(ru\bar{y}) + w(ry\bar{z}) \geq w(ru\bm{\bar{z}}) + w(ry\bm{\bar{y}})$.
Further by the assumptions that $u$ lies on the arc from $y$ to $\bar{z}$ and that $r$ lies on the arc from $\bar{y}$ to $\bar{z}$, we have $\bar{z} \in \wideparen{ru}$.
Therefore, we get that:
\begin{align*}
    \Phi^{\star}(w, u, x, y, z)
    & \geq w(r\bar{u}^2) + w(ru\bm{\bar{z}}) + w(ry\bm{\bar{y}}) + w(rz\bar{r}) \\
    & \qquad
    + \Delta + d(r, u)
    \tag{Quasi-convexity} \\
    & = w(r\bar{u}^2) + w(ru\bar{z}) + w(ry\bar{y}) + w(rz\bar{r}) \\
    & \qquad
    + \Delta + d(r, \bar{z}) + d(\bar{z}, u) 
    \tag{$\bar{z} \in \wideparen{ru}$} \\
    & \geq w(r\bar{u}^2) + w(ru^2) + w(ry\bar{y}) + w(\bar{r}^3)
    \tag{Lipschitzness} \\
    & \geq w(r\bar{u}^2) + w(ru\bm{y}) + w(r\bm{u}\bar{y}) + w(\bar{r}^3)
    \tag{Quasi-convexity} \\
    & = \Phi^{\star}(w, r, u, y, \bar{u})
    ~.
\end{align*}

Finally, consider $u \in \wideparen{yz}$.
Since $r \in \wideparen{\bar{y}\bar{z}}$~~, we have $\bar{r} \in \wideparen{yz}$.
In other words, the antipode of the request $\bar{r}$ partitions arc $\wideparen{yz}$ into two sub-arcs $\wideparen{y\bar{r}}$ and $\wideparen{z\bar{r}}$.
If $u \in \wideparen{y\bar{r}}$, then $w(uy\bar{z})$ can also resolve from $y$ because moving from $u$ to $r$ would reach $y$ midway. 
It then reduces to Subcase 1c.
If $u \in \wideparen{z\bar{r}}$, then $w(uz\bar{r})$ can also resolve from $z$ becuase moving from $u$ to $r$ would reach $z$ midway.
It then reduces to Subcase 1c.

\subparagraph*{Subcase 2b: $w(uy\bar{z})$ resolves from $y$.}
By the assumption of the subcase, moving $y$ towards $r$ would not increase $\Phi^{\star}(w_t, u, x, y, z)$.
We can move $y$ all the way to $\bar{z}$ while maintaining $(x, y, z) \in P$.
Hence we may assume without loss of generality that $y = \bar{z}$.
Then we apply the resolving step and have:
\begin{align*}
    \Phi^{\star}(w, u, x, y, z)
    & = w(r\bar{u}^2) + w(rxz) + w(ru\bar{z}) + w(rz\bar{x}) \\
    & \qquad
    + \Delta + d(r, u) + d(r, \bar{z})
    ~.
    \tag{Resolving}
\end{align*}

If $u \in \wideparen{\bar{x}\bar{z}}$~~, then $(u, z, x) \in P$.
We further have:
\begin{align*}
    \Phi^{\star}(w, u, x, y, z)
    & \geq w(r\bar{u}\bm{z}) + w(rx\bm{\bar{u}}) + w(ru\bar{z}) + w(rz\bar{x}) \\
    & \qquad + \Delta + d(r, u) + d(r, \bar{z})
    \tag{Quasi-convexity} \\
    & \geq w(\bar{r}^3) + w(rx\bar{u}) + w(ru\bar{z}) + w(rz\bar{x})
    \tag{Lipschitznes} \\
    & = \Phi^{\star}(w, r, u, z, x)
    ~.
\end{align*}

If $u \in \wideparen{x\bar{z}}$, then $(u, z, \bar{x}) \in P$.
We further have:
\begin{align*}
    \Phi^{\star}(w, u, x, y, z)
    & \geq w(r\bar{u}\bar{z}) + w(rxz) + w(ru\bar{z}) + w(r\bar{\bar{u}}\bar{x}) \\
    & \qquad
    + \Delta + d(r, u) + d(r, \bar{z})
    \tag{Quasi-convexity} \\
    & \geq w(\bar{r}^3) + w(rxz) + w(ru\bar{z}) + w(r\bar{u}\bar{x})
    \tag{Lipschitzness} \\
    & = \Phi^{\star}(w, r, u, z, \bar{x})
    ~.
\end{align*}

If $x \in \wideparen{xz}$, then $(u, \bar{x}, \bar{z}) \in P$.
The argument depends on the result of applying quasi-convexity to $w(rxz)$ and $w(ru\bar{z})$.
If $w(rxz) + w(ru\bar{z}) \geq w(rx\bm{u}) + w(r\bm{z}\bar{z})$, we have:
\begin{align*}
    \Phi^{\star}(w, u, x, y, z)
    & \geq w(r\bar{u}^2) + w(rx\bm{u}) + w(r\bm{z}\bar{z}) + w(rz\bar{x}) \\
    & \qquad + \Delta + d(r, u) + d(r, \bar{z})
    \tag{Quasi-convexity} \\
    & \geq w(r\bar{u}\bm{z}) + w(rxu) + w(r\bm{\bar{u}}\bar{z}) + w(rz\bar{x}) \\
    & \qquad + \Delta + d(r, u) + d(r, \bar{z})
    \tag{Quasi-convexity} \\
    & \geq w(\bar{r}^3) + w(rxu) + w(r\bar{u}\bar{z}) + w(rz\bar{x})
    \tag{Lipschitzness} \\
    & = \Phi^{\star}(w, r, u, \bar{x}, \bar{z})
    ~.
\end{align*}

Otherwise, we have $w(rxz) + w(ru\bar{z}) \geq w(rx\bm{\bar{z}}) + w(ru\bm{z})$.
We then get that:
\begin{align*}
    \Phi^{\star}(w, u, x, y, z)
    & \geq w(r\bar{u}^2) + w(rx\bm{\bar{z}}) + w(ru\bm{z}) + w(rz\bar{x}) \\
    & \qquad
    + \Delta + d(r, u) + d(r, \bar{z})
    \tag{Quasi-convexity} \\
    & \geq w(r\bar{u}\bm{z}) + w(rx\bar{z}) + w(ruz) + w(r\bm{\bar{u}}\bar{x}) \\
    & \qquad
    + \Delta + d(r, u) + d(r, \bar{z})
    \tag{Quasi-convexity} \\
    & \geq w(\bar{r}^3) + w(rx\bar{z}) + w(ruz) + w(r\bar{u}\bar{x})
    \tag{Lipschitzness} \\
    & = \Phi^{\star}(w, r, u, \bar{z}, \bar{x})
    ~.
\end{align*}

Finally, if $x \in \wideparen{z\bar{x}}$, then $(u, x, \bar{z}) \in P$.
Again, the argument depends on the result of applying quasi-convexity, to $w(ru\bar{z})$ and $w(rz\bar{x})$.
If $w(ru\bar{z}) + w(rz\bar{x}) \geq w(ru\bm{\bar{x}}) + w(rz\bm{\bar{z}})$, we have:
\begin{align*}
    \Phi^{\star}(w, u, x, y, z)
    & \geq w(r\bar{u}^2) + w(rxz) + w(ru\bm{\bar{x}}) + w(rz\bm{\bar{z}}) \\
    & \qquad
    + \Delta + d(r, u) + d(r, \bar{z})
    \tag{Quasi-convexity} \\
    & \geq w(r\bar{u}\bm{z}) + w(rxz) + w(ru\bar{x}) + w(r\bm{\bar{u}}\bar{z}) \\
    & \qquad
    + \Delta + d(r, u) + d(r, \bar{z})
    \tag{Quasi-convexity} \\
    & \geq w(\bar{r}^3) + w(rxz) + w(ru\bar{x}) + w(r\bar{u}\bar{z})
    \tag{Lipschitzness} \\
    & = \Phi^{\star}(w, r, u, x, \bar{z})
    ~.
\end{align*}

Otherwise, we have $w(ru\bar{z}) + w(rz\bar{x}) \geq w(ru\bm{z}) + w(r\bar{x}\bm{\bar{z}})$.
We then get that:
\begin{align*}
    \Phi^{\star}(w, u, x, y, z)
    & \geq w(r\bar{u}^2) + w(rxz) + w(ru\bm{z}) + w(r\bar{x}\bm{\bar{z}}) \\
    & \qquad
    + \Delta + d(r, u) + d(r, \bar{z})
    \tag{Quasi-convexity} \\
    & \geq w(r\bar{u}\bm{z}) + w(rx\bm{\bar{u}}) + w(ruz) + w(r\bar{x}\bar{z}) \\
    & \qquad + \Delta + d(r, u) + d(r, \bar{z})
    \tag{Quasi-convexity} \\
    & \geq w(\bar{r}^3) + w(rx\bar{u}) + w(ruz) + w(r\bar{x}\bar{z})
    \tag{Lipschitzness} \\
    & = \Phi^{\star}(w, r, u, \bar{z}, x)
    ~.
\end{align*}

\subparagraph*{Subcase 2c: $w(uz\bar{x})$ resolves form $z$.}
By the assumption of the subcase, moving $z$ towards $r$ would not increase $\Phi^{\star}(w_t, r, x, y, z)$.
We can move $z$ until it reaches $\bar{y}$ while maintining $(x, y, z) \in P$.
Hence, we may assume without loss of generality that $z = \bar{y}$.
We then apply the resolving step to have have:
\begin{align*}
    \Phi^{\star}(w, u, x, y, \bar{y}) & = w(\bar{u}^3) + w(ux\bar{y}) + w(uy^2) + w(u\bar{x}\bar{y}) \\
    & = w(r\bar{u}^2) + w(rx\bar{y}) + w(ry^2) + w(ru\bar{x}) \\
    & \qquad 
    + \Delta + d(r, u) + d(r, \bar{y})
    ~.
    \tag{Resolving}
\end{align*}

If $u \in \wideparen{\bar{x}\bar{y}}$~~, then $(u, x, y) \in P$.
We further have:
\begin{align*}
    \Phi^{\star}(w, u, x, y, z)
    & \geq w(r\bar{u}\bm{y}) + w(rx\bar{y}) + w(ry\bm{\bar{u}}) + w(ru\bar{x}) \\
    & \qquad + \Delta + d(r, u) + d(r, \bar{y})
    \tag{Quasi-convexity} \\
    & \geq w(\bar{r}^3) + w(rx\bar{y}) + w(ry\bar{u}) + w(ru\bar{x})
    \tag{Lipschitzness} \\
    & = \Phi^{\star}(w, r, u, x, y)
    ~.
\end{align*}

If $u \in \wideparen{x\bar{y}}$, then $(u, \bar{x}, y) \in P$.
The argumenet depends on the result of applying quasi-convexity to $w(rx\bar{y})$ and $w(ru\bar{x})$.
If $w(rx\bar{y}) + w(ru\bar{x}) \geq w_t(rx\bm{u}) + w(r\bm{\bar{y}}\bar{x})$, then:
\begin{align*}
    \Phi^{\star}(w, u, x, y, z)
    & \geq w(r\bar{u}^2) + w(rx\bm{u}) + w(ry^2) + w(r\bar{y}\bm{\bar{x}}) \\
    & \qquad
    + \Delta + d(r, u) + d(r, \bar{y})
    \tag{Quasi-convexity} \\
    & \geq w(r\bar{u}\bm{y}) + w(rxu) + w(ry\bm{\bar{u}}) + w_t(r\bar{y}\bm{\bar{x}}) \\
    & \qquad
    + \Delta + d(r, u) + d(r, \bar{y})
    \tag{Quasi-convexity} \\
    & \geq w(\bar{r}^3) + w(rxu) + w(ry\bar{u}) + w(r\bar{y}\bar{x})
    \tag{Lipschitzness} \\
    & = \Phi^{\star}(w, r, u, \bar{x}, y)
    ~.
\end{align*}

Otherwise, we have $w(rx\bar{y}) + w(ru\bar{x}) \geq w(rx\bm{\bar{x}}) + w(ru\bm{\bar{y}})$.
We then get that:
\begin{align*}
    \Phi^{\star}(w, u, x, y, z)
    & \geq w(r\bar{u}^2) + w(rx\bm{\bar{x})} + w(ry^2) + w(ru\bm{\bar{y}}) \\
    & \qquad 
    + \Delta + d(r, u) + d(r, \bar{y})
    \tag{Quasi-convexity} \\
    & \geq w(r\bar{u}^2) + w(rx\bm{y}) + w(ry\bm{\bar{x}}) + w(ru\bar{y})
    + \Delta + d(r, u) + d(r, \bar{y})
    \tag{Quasi-convexity} \\
    & \geq w(r\bar{u}\bm{y}) + w(rxy) + w(r\bm{\bar{u}}\bar{x}) + w(ru\bar{y})
    + \Delta + d(r, u) + d(r, \bar{y})
    \tag{Quasi-convexity} \\
    & \geq w(\bar{r}^3) + w(rxy) + w(r\bar{u}\bar{x}) + w(ru\bar{y})
    \tag{Lipschitzness} \\
    & = \Phi^{\star}(w, r, u, y, \bar{x})
    ~.
\end{align*}

If $u \in \wideparen{y\bar{x}}$, then $(u, x, \bar{y}) \in P$.
We further have:
\begin{align*}
    \Phi^{\star}(w, u, x, y, z)
    & \geq w(r\bar{u}^2) + wt(rx\bm{y}) + w(ry\bm{\bar{y}}) + w(ru\bar{x}) \\
    & \qquad + \Delta + d(r, u) + d(r, \bar{y})
    \tag{Quasi-convexity} \\
    & \geq w(r\bar{u}\bm{y}) + w(rxy) + w(r\bm{\bar{u}}\bar{y}) + w(ru\bar{x}) \\
    & \qquad + \Delta + d(r, u) + d(r, \bar{y})
    \tag{Quasi-convexity} \\
    & \geq w(\bar{r}^3) + w(rxy) + w(r\bar{u}\bar{y}) + w(ru\bar{x})
    \tag{Lipschitzness} \\
    & = \Phi^{\star}(w, r, u, x, \bar{y})
    ~.
\end{align*}

Finally, if $u \in \wideparen{xy}$, then $(u, \bar{y}, \bar{x}) \in P$.
We futher have:
\begin{align*}
    \Phi^{\star}(w, u, x, y, z)
    & \geq w(r\bar{u}^2) + w(rx\bar{y}) + w(ry\bm{\bar{x}}) + w(ru\bm{y}) \\
    & \qquad
    + \Delta + d(r, u) + d(r, \bar{y})
    \tag{Quasi-convexity} \\
    & \geq w(r\bar{u}\bm{y}) + w(rx\bar{y}) + w(r\bm{\bar{u}}\bar{x}) + w(ruy) \\
    & \qquad
    + \Delta + d(r, u) + d(r, \bar{y})
    \tag{Quasi-convexity} \\
    & \geq w(\bar{r}^3) + w(rx\bar{y}) + w(r\bar{u}\bar{x}) + w(ruy)
    \tag{Lipschitzness} \\
    & = \Phi^{\star}(w, r, u, \bar{y}, \bar{x})
    ~.
\end{align*}

\paragraph*{Case 3: none of $w(ux\bar{y}), w(uy\bar{z}), w(uz\bar{x})$ resolve from $u$.}

Suppose that at least one of them resolves from $x$, $y$, or $z$, and at least one of them resolves from $\bar{x}$, $\bar{y}$, or $\bar{z}$.
Then, there will be two configurations that resolve from points other than $u$ and the pair antipodes in the two configurations.
This reduces to Subcase 1a, whose argument does not use the main assumption of Case 1 as remarked earlier.

Then, by symmetry, we may assume without loss of generality that $w(ux\bar{y}), w(uy\bar{z}), w(uz\bar{x})$ resolve from $x, y, z$ respectively.
Further, this assumption and the potential function are both invariant to permuting $x, y, z$ without changing parity.
We may therefore assume without loss of generality that request $r$ is on arc $\bar{y}\bar{z}$.
We then conclude that $(r, y, z) \in P$, and further assume without loss of generality that $x = r$ because:
\begin{align*}
    \Phi^{\star}(w, u, x, y, z) & = w(\bar{u}^3) + w(ux\bar{y}) + w(uy\bar{z}) + w(uz\bar{x}) \\
    & = w(\bar{u}^3) + w(ur\bar{y}) + w(uy\bar{z}) + w(uz\bar{x}) + d(r, x)
    \tag{Resolving} \\
    & \geq w_t(\bar{u}^3) + w_t(ur\bar{y}) + w_t(uy\bar{z}) + w_t(uz\bar{r})
    \tag{Lipschitzness} \\
    & = \Phi^{\star}(w, u, r, y, z)
    ~.
\end{align*}

With $x = r$, we now have $\bar{z} \in \wideparen{xy} = \wideparen{ry}$.
Hence, $w(uy\bar{z})$ can also resolve from $\bar{z}$.
It then reduces to Subcase 1a.

\section{Canonical Forms of Existing Potential Functions}
\label{app:potential}

\subsection{On Restricting \texorpdfstring{$u = r_t$}{u=r} Explicitly}

Some potential functions in the literature (e.g., Examples~\ref{exa:chrobak-larmore} and \ref{exa:BeinCL}) explicitly restrict $u = r_t$, where recall that $r_t$ is the last request of the work function $w_t$ under consideration.
In other words, instead of defining $\Phi_t = \min_{u, X} \Phi(w_t, u, X)$, they instead let the potential value be:
\begin{equation}
    \label{eqn:r-canonical-potential}
    \Phi_t = \min_{X} \Phi(w_t, r_t, X)
    ~.
\end{equation}

We show that they are equivalent to the canonical form, as long as the difference between two consecutive potential values upper bounds the extended cost as in the canonical argument.

\begin{lemma}
    Suppose that a canonical potential function $\Phi$ in the form of Eqn.~\eqref{eqn:canonical-potential} ensures that for any $k$-server instance and any $t$, the potential values in Eqn.~\eqref{eqn:r-canonical-potential} satisfy:
    \[
        \Phi_t - \Phi_{t-1} \geq w_t(\bar{r}_t^k) - w_{t-1}(\bar{r}_t^k)
        ~,
    \]
    then we also have:
    \[
        \Phi_t = \min_{u, X} \Phi(w_t, u, X)
        ~.
    \]
\end{lemma}

\begin{proof}
    Suppose for contrary that there is a $k$-server instance and some round $t$ for which the lemma does not hold.
    Let $u^*, X^*$ be the minimizer of $\Phi(w_t, u, X)$.
    We have:
    \[
        \Phi(w_t, u^*, X^*) < \Phi_t = \min_X \Phi(w_t, r_t, X)
        ~.
    \]
    
    Consider a new instance by letting $r_{t+1} = u^*$.
    By the form of canonical potential function in Eqn.~\eqref{eqn:canonical-potential}, we have:
    \[
        \Phi(w_{t+1}, r_{t+1}, X)
        = w_{t+1}(\bar{r}_{t+1}^3) + \sum_{i=1}^k w_{t+1}(r_{t+1} X_i) - \sum_{\{(i,j),(i',j')\} \in P} d(x_{i,j}, x_{i',j'})
        ~.
    \]
    
    Further by the equality case of monotonicity (Lemma~\ref{lem:work-function}), it can be written as:
    \begin{align*}
        \Phi(w_{t+1}, r_{t+1}, X)
        &
        = w_{t+1}(\bar{r}_{t+1}^3) + \sum_{i=1}^k w_t(r_t X_i) - \sum_{\{(i,j),(i',j')\} \in P} d(x_{i,j}, x_{i',j'}) \\
        &
        = w_{t+1}(\bar{r}_{t+1}^3) - w_t(\bar{r}_{t+1}^3) + \Phi(w_t, r_{t+1}, X)
        ~.
    \end{align*}
    
    Hence, we get that:
    \begin{align*}
        \Phi_{t+1}
        &
        = \min_X \Phi(w_{t+1}, r_{t+1}, X) \\
        &
        = w_{t+1}(\bar{r}_{t+1}^3) - w_t(\bar{r}_{t+1}^3) + \min_X \Phi(w_t, r_{t+1}, X) \\
        &
        = w_{t+1}(\bar{r}_{t+1}^3) - w_t(\bar{r}_{t+1}^3) + \Phi(w_t, r_{t+1}, X^*)
        ~.
        \tag{$r_{t+1} = u^*$}
    \end{align*}
    
    Combining it with the assumption for contrary, we get that:
    \[
        \Phi_{t+1} - \Phi_t < \Phi_{t+1} - \Phi(w_t, r_{t+1}, X^*) = w_{t+1}(\bar{r}_{t+1}^k) - w_t(\bar{r}_{t+1}^k)
        ~,
    \]
    contracting the lemma assumption.
\end{proof}

\subsection{Equivalence of Existing Potentials and Their Canonical Forms}
\label{app:equivalency}

We first give some useful lemmas for enforcing antipode constraints.

\begin{lemma}[Antipodes in Potentials]
    \label{lem:antipodes-in-potentials}
    For any $1$-Lipschitz function $f: \pointset^{2m} \to R$, where $\pointset$ is a metric with antipodes, we have:
    \[
        \min_{X, Y \in \pointset^m} f(X, Y) - \sum_{x \in X}\sum_{y \in Y}d(x, y) = \min_{u \in \pointset}f(u^m, \bar{u}^m) - m^2\Delta
    \]
\end{lemma}

\begin{proof}
    Note that the right-hand-side is a more restrictive minimization, requiring all $x \in X$ be the same point, all $y \in Y$ be the same point, and the two points are antipodes.
    Hence, the left-hand-side is less than or equal to the right-hand-side.
    
    Next we establish the other direction.
    In other words, for any $X, Y \in \pointset^m$ we shall find $u \in \pointset$ such that:
    \[
        f(u^m, \bar{u}^m) - f(X,Y) \le m^2\Delta - \sum_{x \in X} \sum_{y \in Y} = \sum_{x \in X} \sum_{y \in Y} d(x, \bar{y})
        ~.
    \]
    
    In fact, we will show that letting $u$ be one of the points in $X$ suffices, by proving that:
    \[
        \sum_{u \in X} \Big( f(u^m, \bar{u}^m) - f(X,Y) \Big) \le m \sum_{x \in X} \sum_{y \in Y} d(x, \bar{y})
        ~.
    \]
    
    By the $1$-Lipchitzness of $f$, we have:
    \[
        \sum_{u \in X} \Big( f(u^m, \bar{u}^m) - f(X, Y) \Big) \leq \sum_{x \in X} \sum_{x' \in X: x' \ne x}d(x, x') + \sum_{x \in X}\sum_{y \in Y}d(x, \bar{y})
        ~.
    \]
    
    It remains to bound the first summation on the right by $(m-1)$ times the second summation.
    Further any $y \in Y$ we have:
    \begin{align*}
        \sum_{x \in X} \sum_{x' \in X: x' \ne x} d(x, x')
        &
        \leq \sum_{x \in X} \sum_{x' \in X: x' \neq x} \Big( d(x, \bar{y}) + d(x', \bar{y}) \Big) \\
        &
        = 2(m - 1) \sum_{x \in X} d(x, \bar{y})
        ~.
    \end{align*}
    
    Averaging over $y \in Y$, we have:
    \[
        \sum_{x \in X} \sum_{x' \in X: x' \ne x} d(x, x') \leq \frac{2(m - 1)}{m} \sum_{x \in X} \sum_{y \in Y} d(x, \bar{y}) \le (m-1) \sum_{x \in X} \sum_{y \in Y} d(x, \bar{y})
        ~.
    \]
\end{proof}

\begin{corollary}
    \label{cor:antipodes-in-potentials}
    For any $1$-Lipschitz function $f: (\pointset^{m_1}, \pointset^{m_1}, \dots, \pointset^{m_\ell}, \pointset^{m_\ell}) \to R$, where $\pointset$ is a metric with antipodes, we have:
    \begin{align*}
        &
        \min_{X_i, Y_i \in \pointset^{m_i}} f(X_1, Y_1, \dots, X_\ell, Y_\ell) - \sum_{i=1}^\ell \sum_{x \in X_i} \sum_{y \in Y_i} d(x, y) \\
        & \qquad
        = \min_{(u_1, \dots, u_\ell) \in \pointset^\ell} f(u_1^{m_1}, \bar{u}_1^{m_1}, \dots u_\ell^{m_\ell}, \bar{u}_\ell^{m_\ell}) - \Delta\sum_{i=1}^\ell m_i^2
        ~.
    \end{align*}
\end{corollary}

\paragraph*{Example \ref{exa:this-paper}.}

Recall the canonical form:
\begin{align*}
    \Phi(w_t, u, X)
    &
    = w_t(\bar{u}^3) + w_t(uX_1) + w_t(uX_2) + w_t(uX_3) \\
    & \qquad
    - d(x_{12}, x_{21}) - d(x_{22}, x_{31}) - d(x_{32}, x_{11}) \tag{Antipode constraints} \\
    & \qquad
    - d(x_{11}, x_{21}) - d(x_{21}, x_{31}) - d(x_{31}, x_{11}) \tag{Not on a semi-circle}
    ~.
\end{align*}
Since the work functions are Lipchitzness, and because each of $x_{12}$, $x_{22}$ and $x_{32}$ only appears in one configuration, it is equivalent to
\begin{align*}
    \Phi(w_t) & = \min_{u, x_{11}, x_{21}, x_{31} \in \pointset} w_t(\bar{u}^3) + w_t(ux_{11}\bar{x}_{21}) + w_t(ux_{21}\bar{x}_{31}) + w_t(ux_{31}\bar{x}_{11}) \\
    & \qquad - d(x_{11}, x_{21}) - d(x_{21}, x_{31}) - d(x_{31}, x_{11})
    \\
    & = \min_{u, x, y, z \in \pointset} w_t(\bar{u}^3) + w_t(ux\bar{y}) + w_t(uy\bar{z}) + w_t(uz\bar{x}) - d(x, y) - d(y, z) - d(z, x)
\end{align*}

\paragraph*{Examples~\ref{exa:chrobak-larmore} and \ref{exa:BeinCL}.}

These two cases are evident from the discussion in the main text and from Lemma~\ref{lem:antipodes-in-potentials}.


\paragraph*{Example~\ref{exa:Coester-Koutsoupias}.}

Recall the potential function
\[
    \Phi^\text{\rm CK} (w, x_1, x_2, \dots, x_k) = \sum_{i = 1}^{k+1} w \big( \bar{x}_{i-1}^{i-1} x_i x_{i+1}\dots x_k \big)
    ~.
\]
and the corresponding canonical potential
\[
    \Phi^\text{\rm CK}_\text{\rm Canonical} (w_t, u, X) = w_t(\bar{u}^k) + \sum_{i=1}^k w_t(uX_i) - \sum_{i=1}^k \sum_{j=i+1}^k \sum_{\ell=i}^{k-1} d(x_{i\ell}, x_{ji})
    ~.
\]

Applying Corollary~\ref{cor:antipodes-in-potentials}, let $X_i = x_{i,i}\dots x_{i,k-1}, Y_i = x_{i+1,i}\dots x_{k,i}$, there exists a minimizer such that $x_{i,\ell} = \bar{u}_i$ and $x_{j,i} = u_i$ for any $\ell \in [i, k-1], j \in [i+1, k]$, so we have
\begin{align*}
    \Phi^\text{\rm CK}_\text{\rm Canonical}(w_t) & = \min_{u, u_i \in \pointset}w(\bar{u}^k) + \sum_{i=1}^k w(u u_1 \dots u_{i-1} \bar{u}_i^{k-i}) \\
    & = \min_{x_k, x_{k-1}, \dots, x_1 \in \pointset}w(\bar{x}_k^k) + \sum_{i=1}^k w(x_k x_{k-1} \dots x_{k-i+1} \bar{x}_{k-i}^{k-i}) \\
    & = \min_{x_1, \dots, x_k \in \pointset} \sum_{i = 1}^{k+1} w \big( \bar{x}_{i-1}^{i-1} x_i x_{i+1}\dots x_k \big)
\end{align*}

\end{document}